\documentclass[10pt, a4paper,reqno]{article}

\usepackage[english]{babel}
\usepackage[T1]{fontenc}  
\usepackage[latin1]{inputenc}   

\usepackage{amsmath,amssymb,amsthm,amsfonts}
\usepackage{url}
\usepackage{nicefrac,enumerate,bbm}
\usepackage{stmaryrd} 

\theoremstyle{plain}
\newtheorem{theorem}{{Theorem}}[section] 
\newtheorem*{theorem*}{{Theorem}}
\newtheorem{proposition}[theorem]{Proposition}
\newtheorem*{proposition*}{Proposition}
\newtheorem{lemma}[theorem]{Lemma}
\newtheorem*{lemma*}{Lemma}
\newtheorem{corollary}[theorem]{Corollary}
\newtheorem*{corollary*}{Corollary*}

\theoremstyle{definition}

\newtheorem*{definition*}{Definition}

\theoremstyle{remark}
\newtheorem*{remark*}{Remark}

\makeatletter

\@addtoreset{equation}{section}  
\makeatother

\newcommand{\singl}[1]{\left\{ #1 \right\}}		
\newcommand{\Ii}[2]{\llbracket #1,#2 \rrbracket}	
\newcommand{\R}{\mathbb{R}}

\newcommand{\C}{\mathbb{C}}
\newcommand{\N}{\mathbb{N}}

\newcommand{\tqe}{:}

\renewcommand{\leq}{\leqslant}
\renewcommand{\geq}{\geqslant}
\renewcommand{\bar}[1]{\overline{#1}}
\newcommand{\inv}{^{-1}}
\newcommand {\limt}[2]{\xrightarrow[#1 \to #2]{}}
\newcommand{\abs}[1]{\left\vert #1\right\vert} 
\newcommand{\abspetit}[1]{\vert #1\vert}
\newcommand{\nr}[1]{\left\Vert #1\right\Vert}         
\newcommand{\innp}[2]{\left< #1 , #2 \right>}         
\newcommand{\Four}{\mathcal{F}}		
\newcommand{\Op}{{\mathop{Op}}_h}		
\newcommand{\Opw}{{\mathop{Op}}_h^w}		
\newcommand{\pppg}[1] {\left< #1 \right>} 	
\newcommand{\bigo}[2]{\mathop{O}\limits_{#1 \to #2}}
\newcommand{\littleo}[2]{\mathop{o}\limits_{#1 \to #2}}
\newcommand{\restr}[2]{\left.#1\right|_{#2}}         
\renewcommand{\Re}{\mathop{\rm{Re}}\nolimits}        
\renewcommand{\Im}{\mathop{\rm{Im}}\nolimits}        
\newcommand{\Ran}{\mathop{\rm{Ran}}\nolimits} 
\DeclareMathOperator{\supp}{supp}                    
\newcommand{\1}[1]{\ensuremath{\mathbbm{1}_{#1}}} 

\DeclareMathOperator{\Jac}{Jac}
\DeclareMathOperator{\Hess}{Hess}
\newcommand{\seq}[2]{\left({#1}_{#2}\right)_{#2 \in\N}}
\newcommand{\eqv}{\Longleftrightarrow} 
\newcommand{\trsp}[1]{{\vphantom{#1}}^{\mathit t}{#1}}
\DeclareMathOperator{\Mat}{Mat}
\DeclareMathOperator{\GL}{GL}
\DeclareMathOperator{\sgn}{sgn} 
\newcommand{\symbor}{\Sc_b}	
\newcommand{\symb} {\Sc}
\DeclareMathOperator{\divg}{div}

\renewcommand{\a}{\alpha}
\renewcommand{\b}{\beta}
\newcommand{\g}{\gamma}
\newcommand{\G}{\Gamma}
\renewcommand{\d}{\delta}
\newcommand{\D}{\Delta}
\newcommand{\e}{\varepsilon}

\newcommand{\z}{\zeta}
\newcommand{\y}{\eta}
\renewcommand{\th}{\theta}

\renewcommand{\k}{\kappa}

\renewcommand{\l}{\lambda}
\renewcommand{\L}{\Lambda}
\newcommand{\m}{\mu}
\newcommand{\n}{\nu}
\newcommand{\x}{\xi}
\newcommand{\X}{\Xi}
\newcommand{\s}{\sigma}

\renewcommand{\t}{\tau}
\newcommand{\f}{\varphi}
\newcommand{\vf}{\phi}
\newcommand{\h}{\chi}
\newcommand{\p}{\psi}

\renewcommand{\o}{\omega}
\renewcommand{\O}{\Omega}

\newcommand{\Bc}{{\mathcal B}}

\newcommand{\Gc}{{\mathcal G}}

\newcommand{\Oc}{{\mathcal O}}
\newcommand{\Sc}{{\mathcal S}}
\newcommand{\Tc}{{\mathcal T}}
\newcommand{\Uc}{{\mathcal U}}
\newcommand{\Vc}{{\mathcal V}}
\newcommand{\Wc}{{\mathcal W}}

\newcommand{\hoh}{{H_0^h}}
\newcommand{\huh}{{H_1^h}}
\newcommand{\hh}{H_h}
\newcommand{\negg}{{N_E \G}}
\newcommand{\snegg}{{\tilde\s}}

\begin{document}
\title
{Semiclassical measure for the solution of the dissipative Helmholtz equation\footnotetext{\noindent 2010 \emph{Mathematical Subject Classification.} 35J10, 35Q40, 35S30, 47B44, 47G30, 81Q20.}\footnotetext{\noindent \emph{Keywords.} Dissipative operators, Helmholtz equation, Semiclassical measures.} \footnotetext{This work is partially supported by the French National Research Project NONAa, No. ANR-08-BLAN-0228-01, entitled {\em Spectral and microlocal analysis of non selfadjoint operators}.}}
\author{Julien Royer}
\date{}

\maketitle

\begin{abstract}
We study the semiclassical measures for the solution of a dissipative Helmholtz equation with a source term concentrated on a bounded submanifold. The potential is not assumed to be non-trapping, but trapped trajectories have to go through the region where the absorption coefficient is positive. In that case, the solution is microlocally written around any point away from the source as a sum (finite or infinite) of lagragian distributions. Moreover we prove and use the fact that the outgoing solution of the dissipative Helmholtz equation is microlocally zero in the incoming region.
\end{abstract}

\tableofcontents

\section{Introduction and statement of the result}

We consider on $L^2(\R^n)$ the dissipative semiclassical Helmholtz equation:
\begin{equation} \label{helmholtz} 
(-h^2 \D + V_h - E_h) u_h = S_h
\end{equation}
in the high frequency limit, that is when the semiclassical parameter $h>0$ goes to 0. Here the potential $V_h = V_1 -ihV_2$ has a nonpositive imaginary part of size $h$. We recall (see \cite{benamou-al-03}) that this equation modelizes for instance the propagation of the electromagnetic field of a laser in material medium. In this setting the parameter $h$ is the wave length of the laser, $\Re(E_h - V_h)$ is linked to the electronic density of the material medium (and plays the role of the refraction index for the corresponding hamiltonian problem) while $h\inv \Im(E_h - V_h)$ is the absorption coefficient of the laser energy by the material.
 
Thus, in order to consider the case of a non-constant absorption coefficient we have to allow non-real potentials. We proved in \cite{royer} that if the potential has non-positive imaginary part then (with decay and regularity assumptions on $V$) the resolvent $(-h^2\D + V_h -z)\inv$ is well-defined for $\Im z > 0$ and is of size $O(h\inv)$ uniformly for $z$ close to $E \in \R_+^*$ on condition $E$ satisfies an assumption on classical trajectories for the corresponding hamiltonian problem. In this case, the resolvent has a limit for $z \to E$ in the space of bounded operators in some weighted spaces, and this limit operator gives the (outgoing) solution for \eqref{helmholtz} (see below).\\

Given a source term $S_h$ and such an energy $E > 0$, our purpose in this paper is to study the asymptotic when $h \to 0$ for the outgoing solution $u_h$ of \eqref{helmholtz}. More precisely we are interested in the semiclassical measures (or Wigner measures) of $u_h$. The first work in this direction seems to be the paper of J.-D. Benamou, F. Castella, T. Katsaounis and B. Perthame (\cite{benamou-al-02}). In their paper $S_h = S(x/h)/h$ concentrates on 0 and $\Im E_h = h \a_h$ with $\a_h \to \a \geq 0$. They consider the family of Wigner transforms $f_h$ of the solutions $u_h$ and prove that after extracting a subsequence, this family of Wigner transforms converges to a measure $f$ which is the (outgoing) solution of the transport equation\footnote{ given with our notations.}:
\begin{equation} \label{liouville-cast}
\a f + \x.\partial_x f(x,\x) - \frac 1 2 \partial_x V_1(x) .\partial_\x f(x,\x) = \frac 1 {(4\pi)^2} \d(x) \big|\hat S(\x)\big|^2 \d(\abs \x = 1)
\end{equation}
Note that the solution is estimated by Morrey-Companato-type estimates (see \cite{perthamev99}) and that part of the result is left as a conjecture and proved in \cite{castella05}.

F. Castella, B. Perthame and O. Runborg study in \cite{castella-pr-02} the similar problem with a source term which concentrates on an unbounded submanifold of $\R^n$. As a consequence there is a lack of decay of the source and Morrey-Companato estimates cannot be used. Actually only a formal description of the asymptotics is given and the proof concerns the case where the refraction index is constant, that is $V_1 = 0$, and the submanifold is an affine subspace. X.-P. Wang and P. Zhang give a proof for $V_1\neq 0$ (variable refraction index) in \cite{wangz06} using uniform estimates given by Mourre method. We also mention the work of E. Fouassier who considered the case of a source which concentrates on two points (see \cite{fouassier06}, $V_1=0$ in this case) and the case of a potential discontinuous along an affine hyperplane in \cite{fouassier07} (the source concentrates on 0 in this case). All this papers use \emph{a priori} estimates of the solution in Besov spaces (we have already mentionned \cite{perthamev99}, see also \cite{castellaj06,wangz06,wang07,castellajk08} for further results about these estimates).

Here we are going to use the point of view of J.-F. Bony (see \cite{bony}). He considers the case of a source which concentrates on one or two points (with $V_1 \neq 0$) using a time-dependant method based on a BKW approximation of the propagator to prove that, microlocally, the solution of the Helmholtz equation is a finite sum of lagrangian distributions. In particular, abstract estimates of the solution are only used for the large times control, and this part of the solution has no contribution for the semiclassical measure, so the measure is actually constructed explicitely. Moreover, this method requires a geometrical assumption weaker than the Virial hypothesis used in the previous works. \\

In this paper we consider the case where not only the refraction index but also the absorption coefficient can be non-constant, and hence we have to work with a non-selfadjoint Schrödinger operator. But, as already mentionned, we know that the resolvent is well-defined for a spectral parameter $z$ with $\Im z > 0$. For the selfadjoint semiclassical Schrödinger, we need a non-trapping condition on classical trajectories of energy $E > 0$ to have uniform estimates of the resolvent and the limiting absorption principle around $E$ (see \cite{robertt87,wang87}). In the dissipative case, this assumption can be weakened as follows: any trajectory should either go to infinity or meet the region where $V_2 >0$. This is the assumption we are going to use, and as as consequence, even if we can show that the outgoing solution $u_h$ of \eqref{helmholtz} is microlocally zero in the incoming region, the contribution of large times in $u_h$ does not vanishes when $h\to 0$ as is the case in \cite{bony}, and in particular the solution can be an infinite sum of lagrangian distributions around some points of the phase space. However, the assumption that bounded trajectories should meet the region where there is absorption will make the series of amplitudes of these distributions convergent, which is the key argument in order to have a well-defined semiclassical measure in our case.

Concerning the source term, $S_h$ is allowed to concentrate on any bounded submanifold of $\R^n$. We do not have problem like in \cite{castella-pr-02, wangz06} with decay assumptions, but this allows us to see what happens when the source concentrates on a non-flat submanifold. Note that we do not have phase factor in our source term (see below) so we are in the propagative regime described in \cite{castella-pr-02}.\\

Let us now state the assumptions we are going to use in this work. We denote the free laplacian $-h^2 \D$ by $\hoh$ and $\hh$ is the dissipative Schrödinger operator on $L^2(\R^n)$ ($n\geq 1$): 
\[
\hh = -h^2 \D + V_1(x) -ihV_2(x)
\]
We also denote by $\huh = -h^2 \D + V_1(x)$ the selfadjoint part of $\hh$. $V_1,V_2$ are smooth real functions on $\R^n$, $V_2$ is nonnegative and for $j \in \{1,2\}$, $\a \in\R^n$:
\begin{equation} \label{h1}
\abs{\partial^\a V_j(x)} \leq C_\a \pppg x ^{-\rho - \abs \a}
\end{equation}
for some $\rho > 0$. Here $\pppg \cdot$ denotes the function $x \mapsto (1 + \abs x ^2)^{\frac 12}$. Let $p : (x,\x) \mapsto \x^2 +V_1(x)$ be the symbol on $\R^{2n} \simeq T^*\R^n$ of the selfadjoint part $\huh$. The classical trajectories for this problem are the solutions $\vf^t(w) = (\bar x (t,w),\bar \x(t,w))$ for $w\in\R^{2n}$ of the hamiltonian problem:
\[
\begin{cases}
\partial_t \bar x(t,w) = 2 \bar \x(t,w) \\
\partial_t \bar \x(t,w) = -\nabla V_1 (\bar x(t,w)) \\
\vf^0(w)  = w
\end{cases}
\]

We recall from \cite{royer} that the exact hypothesis we need on an energy $E>0$ to have the limiting absorption principle around $E$ is the following: if we set
\[
\Oc = \singl{x \in \R^n \tqe V_2 (x) > 0}
\]
then for all $w\in\R^{2n}$ such that $p(w) = E$ we have:
\begin{equation}\label{hyp2}
\singl{\vf^t(w),t\in\R} \text{ is unbounded in $\R^{2n}$ or }\singl{\vf^t(w),t\in\R} \cap \Oc \neq \emptyset
\end{equation}
which means that any trapped trajectories should meet the set where there is absorption. For further use we also set, for $\g >0$:
\[
\Oc_\g = \singl{x \in \R^n \tqe V_2 (x) > \g}
\]
With this condition (which is actually necessary), for any $\a > \frac 12$ there exist $\e>0$ and $c\geq 0$ such that:
\[
\sup_{\abs{\Re z - E} \leq \e, \Im z > 0} \nr{\pppg x^{-\a} (\hh - z)\inv \pppg x ^{-\a}} \leq \frac c h
\]
and for all $\l \in [E-\e,E+ \e]$ the limit:
\[
(\hh - (E+i0))\inv := \lim_{\m \to 0^+} (\hh - (E+i\m))\inv
\]
exists (and is a continuous function of $\l$) in the space of bounded operators from $L^{2,\a}(\R^n)$ to $L^{2,-\a}(\R^n)$, where $L^{2,\d}(\R^n)$ stands for $L^2(\pppg x^{2\d} dx)$. Then for all $S_h \in L^{2,\a}(\R^n)$, $u_h = (\hh-(E+i0))\inv S_h \in L^{2,-\a}(\R^n)$ is the outgoing solution for \eqref{helmholtz}.

About the classical hamiltonian problem, we use the following notations:
\begin{eqnarray*}
&&\O_b^\pm(J) = \singl{w \in \R^{2n} \tqe \{\bar x (t,w),\pm t \geq 0\} \text{ is bounded}}\\
&&\O_\infty^\pm(J) = \singl{w \in \R^{2n} \tqe \abs{\bar x (t,w)} \limt t {\pm \infty} +\infty}
\end{eqnarray*}
Note that $\O_\infty^\pm(J)$ is open if $J$ is open and $\O_b^\pm(J)$ is closed if $J$ is closed.\\

Let us now introduce the source term we consider. Given a (bounded) submanifold $\G_2$ of dimension $d \in \Ii 0 {n-1}$ in $\R^n$ with the measure $\s$ induced by the Lebesgue measure on $\R^n$, a smooth function $A$ of compact support on $\G_2$ and a Schwartz function $S \in \Sc(\R^ n)$, we note for $x \in \R^n$:
\begin{equation}
S_h(x) = h^{\frac {1-n-d} 2} \int_{z\in\G} A(z) S\left( \frac {x-z}h \right)\, d\s(z)
\end{equation}

We can choose $\G$ and $\G_1$ open in $\G_2$ such that ${\G_0 := \supp A \subset \G}$, $\bar \G \subset \G_1$ and $\bar {\G_1} \subset \G_2$ (if $\G_2$ is compact we can have $\G_0 = \G = \G_1 = \G_2$).

As usual, for $z \in \G_2$ and $\z \in T_z\G_2$ small enough (where $T_z\G_2$ is the tangent space to $\G_2$ at $z$), we denote by $\exp_z(\z)$ the point $c_\z(1)$  where $t \mapsto c_\z(t)$ is the unique geodesic on $\G_2$ with initial conditions $c_\z(0) =z$ and $c_\z'(0) = \z$ (see \cite[§2.86]{gallothl}). On $\G_2$ we define the distance $d_\G$ as usual: for $x,y\in\G_2$, $d_\G(x,y)$ is the infimum of the length of all piecewise $C^1$ curves from $x$ to $y$. For $z \in \G_2$, there exists a neighborhood $\Uc$ of $z$ in $\G_2$ and $\e > 0$ such that for $x,y \in \Uc$ there is a unique geodesic $c$ from $x$ to $y$ of length less than $\e$. And the length of $c$ is $d_\G(x,y)$ (see \cite[§2.C.3]{gallothl}).

We consider a family of energies $E_h \in \C$ for $h \in ]0,1]$. We assume that $\Im E_h \geq 0$ and:
\begin{equation} \label{h6}
E_h = E_0 + h E_1 + \littleo h 0 (h)
\end{equation}
where $E_0 > 0$ satisfies \eqref{hyp2} and:
\begin{equation} \label{hyp3}
\forall z \in \bar \G, \quad V_1(z) < E_0
\end{equation}

We set $ N\G = \singl{(z,\x) \in \G \times \R^n \tqe \x \bot T_z\G}$,
\begin{equation*} 
\negg = \singl{(z,\x) \in N\G \tqe  \abs{\x} = \sqrt{E_0 - V_1(z)}}
\end{equation*}
and:
\begin{equation*}
\L = \singl{\vf^t(z,\x) ; t >0, (z,\x) \in \negg}
\end{equation*}

We similarly define $\negg_0$ and $\negg _1$. For $(z,\x) \in \negg$ and $(Z,\X) \in T_{(z,\x)}\negg$ we have $Z \in T_z \G$ and $\X \in \R^n$ decomposes as $\X = \X_T + \X_\sslash + \X_\bot$ with $\X_T \in T_z \G$, $\X_\sslash \in \R \x$ and $\X_\bot \in (T_z\G\oplus \R \x)^\bot$. Then $\negg$ is endowed with the metric $g$ defined by:
\[
g_{(z,\x)}\left((Z^1,\Xi^1),(Z^2,\Xi^2)\right) = \innp{Z^1}{Z^2}_{\R^n} + \innp{\X^1_\bot}{\Xi^2_\bot}_{\R^n}
\]
for all $(Z^1,\Xi^1),(Z^2,\Xi^2) \in T_{(z,\x)}\negg$. This means that we do not take into account the part of $\X$ colinear to $\x$ and $T_z \G$, which is allowed since $(Z,\X)$ never reduces to $(0,\X_T + \X_\sslash)$ unless $(Z,\X) = (0,0)$. Indeed, if $Z = 0$ then $\Xi \in T_{(z,\x)} (\negg \cap N_z\G)$ and hence $\Xi = \Xi_\bot$. Now we denote by $\snegg$ the canonical measure on $\negg$ given by the metric $g$. This means that for any smooth map $\p : \Uc \to \Vc$ (where $\Uc$ is an open set in $\R^{n-1}$ and $\Vc$ is an open set in $\negg$) and any function $f$ on $\Vc$ we have (see \cite[§3.H]{gallothl}):
\[
\int_\Vc f(v) \, d\snegg(v) = \int_{\Uc} f(\p(u)) \left( \det (g_{\p(u)}(\partial_i \p (u),\partial_j \p(u)))_{1 \leq i,j\leq n-1}\right) ^{\frac 12} \, du
\]

Finally we set:
\begin{equation*}
\Phi_0 = \singl{(z,\x) \in \negg \tqe \exists t > 0, \vf^t(z,\x) \in \negg}
\end{equation*}
The last assumption we need is:
\begin{equation} \label{hyp4}
\snegg(\Phi_0) = 0
\end{equation}
In \cite[section 4]{bony} is given an example of what can happen without an hypothesis of this kind. Note that when $\G = \singl{0}$, this assumption is weaker than the assumption $\n_0 (E_0-V_1(x)) - x.\nabla V_1(x) \geq c_0 >0$ for some $\n_0 \in ]0,2]$ which is used for instance in \cite{wang07}. This is no longer true in general (for instance we can take $V_1 = 0$, $E_0= 1$ and any circle in $\R^2$ for $\G$).\\

To study semiclassical measures of $u_h$, we choose the point of view of pseudo-differential operators. Let us recall that the Weyl quantization of an observable $a : \R^{2n} \to \C$ is the operator:
\[
\Opw (a) u (x) = \frac 1 {(2\pi h)^n} \int_{\R^n} \int_{\R^n} e^{\frac ih \innp{x-y} \x} a\left(\frac{x+y}2,\x \right) u(y)\, dy \, d\x
\]
We also use the standard quantization:
\[
\Op (a) u (x) = \frac 1 {(2\pi h)^n} \int_{\R^n} \int_{\R^n} e^{\frac ih \innp{x-y} \x} a(x,\x) u(y)\, dy \, d\x
\]
See \cite{robert, martinez, evansz} for more details about semiclassical pseudo-differential operators, \cite{gerard91} for semiclassical measures. We are going to use the following classes of symbols. For $\d \in\R$ we set:
\begin{equation*}
\Sc_\d = \singl{ a \in C^\infty(\R^{2n}) \tqe \forall \a ,\b \in \N^n, \exists c_{\a,\b},\forall (x,\x) \in \R^{2n}, \abs{\partial_x ^\a \partial_\x^\b a (x,\x)} \leq c_{\a,\b} \pppg x ^{\d-\abs \a}}
\end{equation*}
while $\Sc_b$ is the set of $C^\infty(\R^{2n})$ functions whose derivatives up to any order are in $L^\infty(\R^{2n})$.\\

We can now state the main theorem of this paper:

\begin{theorem} \label{th2.1}
There exists a Radon measure $\m$ on $\R^{2n}$ such that for all $q\in C_0^\infty(\R^{2n})$:
\begin{equation*}
\innp{\Opw(q) u_h}{u_h} \limt h 0  \int_{\R^{2n}} q \, d\m
\end{equation*}

Moreover $\m$ is characterized by the following three properties:
\begin{enumerate}[(i)]
\item $\m$ is supported on the hypersurface of energy $E_0$:
\begin{equation*}
\supp \m \subset p\inv(\singl{E_0})
\end{equation*}
\item $\m$ vanishes in the incoming region: let $\s \in]0,1[$, then there exists $R \geq 0$ such that for $q\in C_0^\infty (\R^{2n})$ supported in the incoming region $\G_-(R,-\s)$ (see definition in section \ref{sec-tps-gd}) we have:
\begin{equation*}
\int q \, d\m = 0 
\end{equation*}
\item $\m$ satisfies the Liouville equation:
\begin{equation} \label{liouville}
(H_p + 2 \Im E_1 + 2 V_2) \m = \pi (2\pi)^{d-n} A(z)^2 \abs \x \inv \hat S (\x)^2   \snegg
\end{equation}
where $H_p = \{p,\cdot\} = 2\x. \partial_x - \nabla V_1(x) .\partial_\x$ and $\snegg$ is extended by 0 on $\R^{2n} \setminus \negg$. This means that for any $q \in C_0^\infty(\R^{2n})$ we have:
\[
\int_{\R^{2n}} (-H_p + 2 \Im E_1 + 2 V_2) q \, d\m = \pi (2\pi)^{d-n}\int_{\negg} q(z,\x) A(z)^2 \abs \x \inv \hat S(\x)^2 \, d\snegg(z,\x)
\]
\end{enumerate}
\end{theorem}

We first remark that this theorem gives not only existence of a semiclassical measure but also uniqueness, since we do not need to extract a subsequence to have convergence of $\innp{\Opw(q) u_h}{u_h}$ when $h \to 0$.

Moreover, we see that in the Liouville equation the absorption coefficient $\a$ of \eqref{liouville-cast} is replaced by our full non-constant absorption coefficient $\Im E_1 + V_2$, as one could expect.

And finally we will prove that the three properties of the theorem implies that the measure $\m$ is given, for $q \in C_0^\infty(\R^{2n})$, by:
\begin{equation} \label{expr-mu}
\int_{\R^{2n}} \! \! q \, d\mu = \pi (2\pi)^{d-n}\!\!\int_{\R_+} \!\int_\negg \!\!\! A(z)^2 \abs \x \inv  \! \hat S (\x)^2  q(\vf^t(z,\x)) e^{-2t\Im E_1 - 2\int_0^t V_2(\bar x (s,z,\x))\,ds} \,d\snegg(z,\x) \,dt
\end{equation}

To prove this theorem we write as in \cite{bony} the resolvent as the integral over positive times of the propagator, the main difference being the large times contribution. Let:
\[
U_h(t) = e^{-\frac{it}h \hh}, \quad U_0^h (t) =e^{-\frac{it}h \hoh},\quad \text{and}\quad U_h^E(t) = e^{-\frac{it}h (\hh-E_h)}
\]
Then:
\begin{equation} \label{uh-int} 
u_h = (\hh -(E_h+i0))\inv S_h = \frac ih \int_0^{+\infty} U_h^E(t) S_h\,dt
\end{equation}
and for $T \geq 0$ we set:
\begin{equation} \label{def-uht}
\begin{aligned}
u_h^T 
& = (\hh -(E_h+i0))\inv S_h-(\hh -(E_h+i0))\inv U_h^E(T) S_h\\
& = \frac ih \int_0^{T} U_h^E(t) S_h\,dt
\end{aligned}
\end{equation}
Our purpose is to study the quantity:
\[
\lim _{h\to 0} \lim_ {T \to +\infty} \innp{\Op(q) u_h^T}{u_h^T}
\]
which we cannot do directly.
Around $w \in \R^{2n}$, troubles appear when proving that relevant parts of integral \eqref{uh-int} are around times $t$ for which we can find $(z,\x)\in\negg$ such that $\vf^t(z,\x) = w$ (see proposition \ref{prop-terme-reste}). Indeed, far from these times we can find $t$ such that $\vf^t (\negg)$ is close to $w$, giving contribution for the semiclassical measure in any neighborhood of $w$. Moreover, the Egorov theorem we use gives estimates uniform in $h$ but not in time (see \cite{bouzouinar02} for a discussion of this problem). The key of our proof is to check that even if the contribution of large times is not zero as for the non-trapping case, the damping term $V_2$ makes it so small that the semiclassical measure is also given by:
\[
\lim _{T \to +\infty} \lim _{h\to 0} \innp{\Op(q) u_h^T}{u_h^T}
\]
which is much easier to study. Indeed, this means that we study the semiclassical measure for the family $(u_h^T)$. This can be done as for the non-trapping case since  we do not have to worry about large times behavior. This gives a family of measures on $\R^{2n}$, and then we can take the limit $T \to +\infty$, since we no longer have problems with the parameter $h$. It only remains to check this gives the measure we are looking for.\\

We begin the proof by a few preliminary results: we show to what extent the damping term $V_2$ implies a decay of $U_h(t)$, we look at the classical trajectories around the submanifold $\G$ and give more details about the assumption on $\Phi_0$. Finally we show that the solution $u_h$ concentrates on the hypersurface of energy $E_0$. In section 3 we give an estimate of the solution near $\G$, since we cannot give a precise description of $u_h$ there. This part is close to section 3.3 of \cite{bony} but we give a complete proof in order to see how to deal with the general case $\dim \G \geq 1$. In section 4 we study the finite times contribution and give the semiclassical measure for $u_h^T$, and then in section 5 we prove that taking the limit $T \to +\infty$ for this family of measures gives a semiclassical measure for the solution $u_h$. We also show that this limit is the solution of the Liouville equation \eqref{liouville} where $V_2$ naturally appears as a damping factor.

Finally in section 6 we give the proof of the estimate in the incoming region we use in section 5. Indeed if we no longer assume that all the classical trajectories of energy $E_0$ go to infinity, there still are some non-trapped trajectories. So we still need the estimate of the outgoing solution in the incoming region used in the non-trapping case. For the self-adjoint Schrödinger operator, this is proved in \cite{robertt89} but here we need to show that this remains true in our dissipative setting.

\section{Some preliminary results}

\subsection{Damping effect of the absorption coefficient on the semigroup generated by $\hh$}

We saw in \cite{royer} that assumption \eqref{hyp2} is actually satisfied for any energy close enough to $E_0$, hence we can consider two closed intervals $I$ and $J$ such that $E_0 \in \mathring I$, $\bar I \subset \mathring J$ and any trapped trajectory of energy in $J$ meets $\Oc$.

The main tool we need in this section is the dissipative version of Egorov theorem. We already stated this theorem in \cite{royer} but we give here a more precise version we are going to use in the proof of proposition \ref{prop-terme-reste}.

\begin{proposition} \label{prop-Egorov}
Let $a \in \symbor$.
\begin{enumerate}[(i)]
\item There exists a family of symbols $\a_j(t)$ for $j\in\N$ and $t\geq 0$ such that for any $N\in\N$ and $t \geq 0$ the symbol $A_N(t,h) = \sum_{j=0}^{N} h^j \a_j(t)$ satisfies:
\[
U_h(t)^* \Opw(a) U_h(t) = \Opw (A_N(t,h)) + \bigo h 0 (h^{N+1})
\]
where the rest is bounded as an operator on $L^2(\R^n)$ uniformly in $t \in [0,T]$ for any $T \geq 0$.
\item $\a_0(t) = (a \circ \vf^{t}) \exp\left( -2 \int_0^t V_2 \circ \vf^s\,ds\right)$ where for $(x,\x) \in \R^{2n}$, $V_2(x,\x)$ means $V_2(x)$.
\item If $a$ vanishes on the open set $\Wc \subset \R^{2n}$ then for all $j\in\N$ the symbol $\a_j(t)$ vanishes on $\vf^{-t}(\Wc)$.
\end{enumerate}
\end{proposition}

\begin{proof}
In \cite{royer} we proved (i) for $N = 0$ and (ii). Moreover (iii) is a direct consequence of (ii) for $j=0$. What remains can be proved as in the selfadjoint case (see \cite{robert}) so we only recall the ideas. (i) is proved by induction. More precisely, we show that for any $N\in\N$:
\begin{eqnarray*}
\lefteqn{U_h(t)^* \Opw (a) U_h(t) = \sum_{j=0}^N h^j \Opw (\a_j(t)) }\\
&& + h^{N+1} \int_{\t_1 = 0}^t \int_{\t_2 = 0}^{\t_1} \dots \int_{\t_{N+1}=0}^{\t_N}  U_h(\t_{N+1})^* \Opw(b_N(\t_1,\dots,\t_{N+1},h)) U_h(\t_{N+1}) \, d\t_{N+1}\dots d\t_1
\end{eqnarray*}
for some symbol $b_N$. The case $N+1$ is obtained by applying the case $N=0$ to the principal symbol of $b_N$.

To prove (iii) we take the derivative of $U_h(t)^* \Opw(a) U_h(t)$ with report to $t$. This gives, for $j\in\N$:
\[
\partial_t \a_j(t) = H_p(\a_j) - 2 V_2 \a_j(t) + \sum_{q= 0}^{j-1} C_{j,q} D^*_{j,q} \a_q 
\]
where $C_{j,q}$ is a function with bounded derivatives and $D^*_{j,q}$ is a differential operator. Then if $\tilde \a_j(t) = (\a_j(t) \circ \vf^{-t}) \exp\left(2 \int_0^t V_2 \circ \vf^{-s} \,ds\right)$ we have:
\[
\partial_t \tilde \a_j(t) = \sum_{q=0}^{j-1} C_{j,q} D^*_{j,q}(\a_q(t) \circ \vf^{-t}) \exp\left(2 \int_0^t V_2 \circ \vf^{-s} \,ds\right)
\]
and it is easy to check by induction on $j\geq 1$:
\[
\tilde \a_j (0) = 0, \quad \partial_t \tilde \a_j(t) = 0 \text{ on } \Wc, \quad  \text{and hence } \a_j(t) = 0 \text{ on } \vf^{-t}(\Wc)
\]
\end{proof}

\begin{lemma} \label{prop-amortis}
Let $K$ be a compact subset of $\O_b^+(J)$. There is $C \geq 0$ and $\d >0$ such that:
\[
\forall w \in K, \quad \exp\left( -\int_{s=0}^t V_2(\vf^s(w))\, ds\right) \leq C e^{-\d t}
\]
\end{lemma}

\begin{proof}
\noindent{\bf 1.}
We first recall that if $w \in \O_b^+(J)$ then there exists $T \geq 0$ such that $\vf^T(w) \in \Oc$ (this is slightly stronger than assumption \eqref{hyp2}). Indeed, the set $K_w = \bar{\singl{\vf^t(w),t \geq 0}}$ is compact, so there is an increasing sequence $\seq t m$ with $t_m \to +\infty$ and $w_\infty \in K_w$ such that $\vf^{t_m}(w) \to w_\infty$. Since $\O_b^+(\singl{p(w)})$ is closed, $w_\infty \in \O_b^+(\singl{p(w)})$. Moreover, for $M \in \N$ and $m \geq M$ we have $\vf^{-t_M}(\vf^{t_m}(w)) \in K_w$ and hence $\vf^{-t_M}(w_\infty) \in K_w$, which proves that $w_\infty \in \O_b^-(\R)$. By assymption \eqref{hyp2}, there is $T \in \R$ such that $\vf^T(w_\infty) \in \Oc$. Hence $\vf^{T+t_m}(w)$ lies in $\Oc$ for large $m$. Since $T+t_m \geq 0$ when $m$ is large enough, the claim is proved.\\

\noindent{\bf 2.}
We set: 
\[
\tilde K = \overline{\singl{\vf^t(w), t \geq 0, w \in K}}
\]
By definition of $K$, $\tilde K$ is compact in $\R^{2n}$. Let $w \in \tilde K$. There are $T_w \geq 0$ and $\g_w >0$ such that $\vf^{T_w}(w) \in \Oc_{2\g_w}$, so we can find $\t_w >0$ and a neighborhood $\Vc_w$ of $w$ in $\R^{2n}$ such that for all $v \in \Vc_w$ and $t \in [T_w-\t_w, T_w ]$ we have: $\vf^t(v) \in \Oc_{\g_w}$. As $\tilde K$ is compact we can find $w_1,\dots,w_k$ such that $K \subset \cup_{i=1}^k \Vc_{w_i}$. Then we take $T = \max \{ T_{w_i}, 1\leq i \leq k\}$, $\t = \min\{ \t_i, 1 \leq i \leq k\}$ and $\g = \min \{ \g_{w_i}, 1 \leq i\leq k\}$. For all $w \in K$ and $t\geq 0$, $\vf^t(w)$ is in $\tilde K$ and hence in $[t,t+T]$ there is a subinterval $I_{w,t}$ of length at least $\t$ such that $\vf^s(w) \in \Oc_\g$ for $s \in I_{w,t}$. Thus:
\[
\exp\left( \int_{s=t}^{t+T} V_2(\vf^s(w))\, ds\right) \leq  e^{-\t \g}
\]
We apply this for $t_n = nT$ with $n \leq t/T$ and this gives:
\[
\begin{aligned}
\exp\left( \int_{0}^{t} V_2(\vf^s(w))\, ds\right)
 \leq e^{- \frac{t-T}T \t\g}
& \leq e^{\t\g} e^{-t\frac {\t\g} T}
\end{aligned}
\]
so the result follows with $C =  e^{\t\g}$ et $\d = \frac {\t\g} T$.
\end{proof}

\begin{proposition}  \label{prop-super-egorov}
Let $q,q' \in C_0^\infty(\R^{2n})$ supported in $p\inv(J)$ and $\e>0$. Then there exists $T_0 \geq 0$ such that for all $T \geq T_0$ we can find $h_T > 0$ which satisfies:
\[
\forall h \in ]0, h_T], \quad \nr{\Opw(q) U_h(T) \Opw(q')} \leq \e
\]
\end{proposition}

\begin{proof}
We set $K = \supp q' \cap \O_b^+(\R)$. As $K$ is a compact subset of $\O_b^+(J)$, lemma \ref{prop-amortis} shows that there is $T_0 \geq 0$ such that:
\[
\sup_{w\in K} \nr {q} _\infty \nr{q'}_\infty \exp\left(-\int_{s=0}^T (V_2 \circ \vf^{s}) (w)\, ds \right) \leq \frac \e 4
\]

As the left-hand side is a continuous function of $w$, we can find a neighborhood $\Vc$ of $K$ in $\R^{2n}$ such that this holds for $w \in \Vc$ after having replaced $\e/4$ by $\e/2$. Let now $K_\infty = \supp q' \setminus  \Vc$. $K_\infty$ is a compact subset of $\O_+^\infty$. Therefore, if $T_0$ is large enough, we can assume that for $T\geq T_0$ and $w \in K_\infty$ we have  $\vf^{T}(w) \notin \supp q$. Hence by Egorov theorem (see also remark 4.4 in \cite{royer}), for any $T\geq T_0$ we have:
\begin{equation}
\begin{aligned}
\nr{\Opw(q) U_h(T) \Opw(q')}
& = \nr{U_1^h(-T) \Opw(q) U_h(T) \Opw(q')}\\
& = \nr{\Opw\left((q\circ \vf^T)e^{-\int_{s=0}^T V_2 \circ \vf^s\, ds} \right) \Opw(q')} + \bigo h 0 (h)\\
& \leq \sup_{w\in\R^{2n}} \abs{q'(w) (q ( \vf^T(w))) e^{-\int_{s=0}^T V_2 ( \vf^s(w)) \, ds}} + C(T) \sqrt h\\
& \leq \frac \e 2 + C(T) \sqrt h
\end{aligned}
\end{equation}
and hence for any fixed $T \geq T_0$ we can find $h_T>0$ small enough to conclude.
\end{proof}

\subsection{Classical trajectories around $\G$}

In this section we assume that assumptions \eqref{h1}, \eqref{hyp2} and \eqref{hyp3} are satisfied.

\begin{proposition} \label{prop-superbij}
There exists $\t_0 > 0$ such that:
\begin{equation}\label{superbij}
\Tc : \begin{cases} ]0,3\t_0] \times \negg_1 & \to  \R^{n} \\ (t,w) & \mapsto  \bar x (t,w) \end{cases}
\end{equation}
is one-to-one and $\Ran (\Tc)\cup \G_1$ is a neighborhood of $\G$ in $\R^n$. Furthermore:
\begin{enumerate}[(i)]
\item We can choose $\t_0$ to have:
\begin{equation}\label{dist-gamma}
\forall t \in ]0,3\t_0], \forall w \in \negg_1, \quad 2 \g_m t \leq d(\bar x (t,w),\G_2) \leq 2 \g_M t
\end{equation}
for some $\g_M \geq \g_m >0$.
\item If $f$ is a continuous function with support in $\Tc(]0,3\t_0[ \times \negg)$ then:
\begin{equation}\label{chgt-var}
\int_{x\in\R^n} f(x)\,dx = 2^{n-d} \int_{0}^{3\t_0} \!\! \int_{\negg} f(\bar x (t,z,\x))  t^{n-d-1} \abs \x \Big(1+\bigo  t 0 (t) \Big)\, d\snegg(z,\x)\,dt 
\end{equation}
\end{enumerate}
\end{proposition}

For $0\leq r_1 \leq r_2 \leq 3\t_0$ we set:
\[
\tilde \G(r_2) = \Tc ([0,r_2] \times \negg) \quad \text{and} \quad \tilde \G(r_1,r_2) = \Tc (]r_1,r_2] \times \negg)
\]
When $x \in \tilde \G(0,3\t_0)$ we write $(t_x,z_x,\x_x) = \Tc\inv(x)$.

\begin{proof}
For $\t > 0$, let :
\[
N(\t) = \singl{(z,\x) \in N\G_1 \tqe  \abs \x \leq \t  \sqrt{E_0-V_1(z)}}
\]

We consider the function $\tilde \Tc$ from $N(1)$ to $\R^n$ defined by:
\[
\tilde \Tc(z,\x) = \begin{cases} \bar x \left(\frac {\abs \x}{ \sqrt{E_0-V_1(z)}} , z ,\frac {\x\sqrt{E_0-V_1(z)}}{\abs \x }\right) & \text{if } \x \neq 0 \\ z & \text{if } \x = 0 \end{cases}
\]
We have:
\[
\tilde \Tc (z,\x) = z + 2\x + o(\abs \x)
\]
Hence for $\t_0 >0$ small enough, $\tilde \Tc$ is a diffeomorphism from $N(3\t_0)$ to a tubular neighborhood of $\G_1$ (we can follow the proof for the function $(x,\x) \mapsto z +2\x$, see for instance theorem 2.7.12 in \cite{berger-gostiaux}). In particular $\tilde \Tc$ and hence $\Tc : (t,z,\x) \mapsto \tilde \Tc (z , t\x)$ are one-to-one and $\Ran \Tc \cup \G_1 = \Ran \restr{\tilde \Tc}{N(3\t _0)} \cup \G_1$ is a neighborhood of $\G_0$.

(i)
We have:
\[
\begin{aligned}
\bar x (t,z,\x) - z
 = \int_{0}^t  2\bar \x (s,z,\x)\,ds 
 = 2t\x - 2\int_{0}^t \int_{0}^s \nabla V_1 (u,z,\x) \, du \,ds
\end{aligned}
\]
Hence, if $M = \sup _{x\in\R^n} \abs{\nabla V_1(x)}$ this gives:
\[
\abs{\bar x (t,z,\x) - z - 2t\x} \leq 2t^2 M
\]
Denote $\x_{\min} = \min \{\abs \x , \x \in \negg_1\}>0$ and $\x_{\max} = \max \{\abs \x , \x \in \negg_1\}$. We recall from \cite{berger-gostiaux} that for $(z,\x) \in \negg_1$ and $t$ small enough we have $d(z+t\x,\G_2) =  t \abs \x$.
Then for $\t_0$ small enough we have $2\t_0 M \leq \x_{\min}$ so:
\[
d(\bar x (t,z,\x),\G_2) \geq d(z+2t\x,\G_2) - \abs{\bar x (t,z,\x) - z - 2t\x} \geq 2 t \abs \x - t {\x_{\min}} \geq t  {\x_{\min}} 
\]
and:
\[
d(\bar x (t,z,\x),\G_2) \leq d(z+2t\x,\G_2) + \abs{\bar x (t,z,\x) - z - 2 t\x} \leq 2 t \abs \x + t \x_{\min} \leq t (2 \x_{\max} + \x_{\min}) 
\]

(ii)
Let $(t,z,\x) \in ]0, 3\t_0[ \times \negg$. For $(T_1,Z_1,\Xi_1),(T_2,Z_2,\Xi_2) \in T_{(t,z,\x)} (]0,3\t_0[\times \negg)$ we set:
\[
\tilde g_{(t,z,\x)}( (T_1,Z_1,\Xi_1),(T_2,Z_2,\Xi_2)) = T_1 T_2 + g_{(z,\x)}((Z_1,\Xi_1),(Z_2,\Xi_2))
\]

We first look for good orthonormal bases of $ T_{(t,z,\x)} (]0,3\t_0[\times \negg)$ (for the metric $\tilde g$) and $\R^n$ (for the usual metric) to compute the jacobian of $\Tc$.
$\negg \cap (\singl z \times \R^n)$ is a submanifold of dimension $n-d-1$ in $\negg$, so we can consider an orthonormal basis $((0,\Xi_j))_{d+2\leq j \leq n}$ of its tangent space at $(z,\x)$. We now choose an orthonormal basis $(Z_j)_{2 \leq j \leq d+1}$ of $T_z\G$. We can find $\Xi_{2},\dots , \Xi_{d+1} \in \R^n$ such that $(Z_j,\Xi_j) \in T_{(z,\x)}\negg$ for $j \in \Ii {2}{d+1}$ and since linear combinations of $(0,\Xi_{d+2}),\dots,(0,\X_{n})$ can be added, we may assume that $\Xi_j \in T_z\G \oplus \R \x$ for all $j \in \Ii {2}{d+1}$. These $n-1$ vectors form an orthonormal family of $T_{(z,\x)}\negg$ to which we add the canonical unit vector of $\R$ for the time component. This gives an orthonormal basis $\Bc_{(t,z,\x)}$ of $T_{(t,z,\x)} (]0,3\t_0[ \times \negg)$. In $\R^n$ we consider the orthonormal basis:
\[
\tilde \Bc_{\Tc(t,z,\x)} = (\x/\abs \x,Z_{n-d},\dots,Z_{n-1},\Xi_1,\dots , \X_{n-d-1})
\]
Since $\Tc(t,z,\x) = z  +2t\x + O(t^2)$, the jacobian matrix of $\Tc$ in these two bases is:
\[
\Mat _{\Bc_{(t,z,\x)} \to \tilde \Bc_{\Tc(t,z,\x)}} D_{(t,z,\x)}\Tc = 
\begin{pmatrix} 
2 \abs \x & 0 & 0 \\ 
0 & I_d & 0 \\ 
0 & 0 & 2t I_{n-d-1}
\end{pmatrix}
 \left( 1 + \bigo t 0 (t) \right) 
\]
On the other hand, since basis $\Bc_{(t,z,\x)}$ and $\tilde \Bc_{\Tc(t,z,\x)}$ are orthonormal, we have, for $x \in \tilde \G(0,3\t_0)$:
\[
\left(\det (\tilde g_{\Tc\inv(x)} (\partial_i \Tc\inv (x),\partial_j \Tc\inv (x)))_{1 \leq i,j\leq n}\right)^{\frac 12} = \abs {\det \Mat _{ \tilde \Bc_{x}\to\Bc_{\Tc\inv(x)} } D_x\Tc\inv }
\]
Thus, using the definition of the measure $dt \, d\snegg$ on $]0,3\t_0[\times \negg$ and the fact that $\Tc \inv : \tilde \G(0,3\t_0) \to ]0,3\t_0[ \times \negg $ can be seen as a map for the manifold $]0,3\t_0[ \times \negg$, we obtain:
\begin{eqnarray*}
\lefteqn{\int_{x \in \R^n} f(x)\, dx}\\
&& = \int_{x\in\R^n} (f\circ \Tc) (\Tc\inv x) \abs {\det \Mat _{ \tilde \Bc_{x}\to\Bc_{\Tc\inv(x)} } D_x\Tc\inv}\abs {\det \Mat _{  \Bc_{\Tc\inv(x)} \to\tilde\Bc_{x}} D_{\Tc\inv(x)}\Tc }\,dx\\
&& = \int_{t=0}^{3\t_0} \int_{(z,\x)\in \negg} (f\circ \Tc)(t,z,\x) \abs {\det \Mat _{  \Bc_{(t,z,\x)} \to\tilde\Bc_{\Tc(t,z,\x)}} D_{(t,z,\x)}\Tc  } \, d\snegg(z,\x) \, dt\\
&& = 2^{n-d} \int_{0}^{3\t_0} \!\! \int_{\negg} f(\Tc(t,z,\x))  t^{n-d-1} \abs \x \Big(1+\bigo  t 0 (t) \Big)\, d\snegg(z,\x)\,dt
\end{eqnarray*}
\end{proof}

\begin{corollary} \label{dist-e0}
Let $(t,z,\x)\neq (s,\z,\y) \in \R_+^* \times \negg$ such that $\vf^t(z,\x) = \vf^s(\z,\y)$. Then $\abs {t-s} \geq 3\t_0$ where $\t_0$ is given by proposition \ref{prop-superbij}.
\end{corollary}

Let $w \in \R^{2n}$ and denote:
\[
( (t_{w,k},z_{w,k},\x_{w,k}))_{1\leq k \leq K_w} = \singl{ (t,z,\x) \in \R_+^* \times \negg \tqe \vf^t(z,\x) = w}
\]
with $t_{w,1} < t_{w,2}<\dots$ and $K_w \in \N \cup \singl \infty$ ($\Ii 1 {K_w}$ is to be understood as $\N^*$ if $K_w = \infty$ and $K_w = 0$ if $w \notin \L$). We also define $K_w^T = \sup \singl{k \in \Ii 1 {K_w} \tqe t_{w,k} \leq T} \in \N$. For $w \in \R^{2n}$ and $k\in \Ii 1 {K_w}$ we write:
\[
\L_{w,k} = \singl{\vf^t(z,\x), \abs{t-t_{w,k}} <\t_0, \abs{(z,\x)-(z_k,\x_k)} < \t_0}
\]
and if $w \in \negg$:
\[
\L_{w,0} = \singl{\vf^t(z,\x), \abs t  < \t_0, \abs{(z,\x)-w} <\t_0}
\]

\begin{proposition}
Let $w = (x,\x) \in \R^{2n}$ and $j,k\in \Ii 1 {K_w}$ ($\Ii 0 {K_w}$ if $w \in \negg$). Then
\begin{enumerate}[(i)]
\item $\L_{w,j} \cap \L_{w,k}$ is of measure zero in $\L_{w,j}$ is and only if it is of measure zero in $\L_{w,k}$.
\item Assumption \eqref{hyp4} is equivalent to:
\begin{equation}\label{hyp4bis}
\forall w \in \R^{2n},\forall j , k \in \Ii 1 {K_w} (\text{or } \Ii 0 {K_w}), \quad \L_j\cap\L_k\text{ is of measure 0 in } \L_j
\end{equation}
\end{enumerate}
\end{proposition}

This proposition is proved in section 6 of \cite{bony}.

\subsection{Localization around $E_0$-energy hypersurface}

\begin{proposition} \label{prop-norme-S}
For any $\d \in \R$ we have:
\begin{equation} \label{norme-S}
\nr{S_h}_{L^{2,\d} (\R^n)} = \bigo h 0 \big(\sqrt h\big)
\end{equation}
\end{proposition}

\begin{proof}
\noindent{\bf 1.} There exists $C\geq 0$ such that for all $x\in\R^n$ and $r>0$, the measure of $B(x,r)\cap \G$ in $\G$ is less than $Cr^d$. Otherwise for all $m\in\N$ we can find $x _m \in \R^n$ and $r _m > 0$ such that the measure of the ball $B(x_m,r_m)\cap \G$ in $\G$ is greater than $mr_m^d$. As $\G$ is of finite measure, $r_m$ necessarily goes to 0 as $m \to +\infty$. On the other hand $x_m$ has to stay close to $\G$, hence in a compact subset of $\R^n$, so taking a subsequence we can assume that $x_m \to x_\infty \in \G$. But the part of $\G$ close to $x_\infty$ is diffeomorphic to a subset of $\R^d \subset \R^n$, hence the measure of $B(x_\infty,r)\cap \G$ in $\G$ is less than $Cr^d$ for some $C\geq0$.\\

\noindent{\bf 2.} Let $x \in \R^n$. We have:
\[
\begin{aligned}
S_h(x)^2
& = h^{1-n-d} \left( \sum_{m\in\N} \int_{mh\leq \abs{x-z}< (m+1)h} A(z)S\left( \frac {x-z}h\right) \, d\s(z)\right)^2\\
& \leq c \, h^{1-n-d}  \sum_{m\in\N}m^2\left( \int_{mh\leq \abs{x-z}< (m+1)h} A(z)S\left( \frac {x-z}h\right) \, d\s(z)\right)^2\\
& \leq c \,h^{1-n}\sum_{m\in\N}m^{2+d} \int_{mh\leq \abs{x-z}< (m+1)h}  S\left( \frac {x-z}h\right)^2 \, d\s(z)\\
\end{aligned}
\]
and hence:
\[
\begin{aligned}
\nr{S_h}^2_{L^{2,\d}(\R^n)}
& \leq c\, h^{1-n} \int_{x\in\R^n}\sum_{m\in\N}m^{2+d} \int_{mh\leq \abs{x-z}< (m+1)h} \pppg x ^{2\d}  S\left( \frac {x-z}h\right)^2 \, d\s(z)\, dx\\
& \leq c\, h \sum_{m\in\N}m^{2+d} \int_{z\in\G} \int_{m\leq \abs{y}< (m+1)} \pppg {z+hy} ^{2\d}  S(y)^2 \, dy \, d\s(z)\\
& \leq c\, h \sum_{m\in\N}m^{2+d} \int_{z\in\G} \int_{m\leq \abs{y}< (m+1)} \pppg {y} ^{2\d}  S(y)^2 \, dy \, d\s(z)
\end{aligned}
\]
for $h\in]0,1]$, since $\G$ is bounded. As $S$ decays faster than $\pppg y ^{-\frac{n+2\d+4+d}2}$ we have:
\[
\begin{aligned}
\nr{S_h}^2_{L^{2,\d}(\R^n)}
& \leq c\, h \sum_{m\in\N}m^2 \pppg m ^{-4-d} \leq c\, h
\end{aligned}
\]
\end{proof}

Since $(\hh - (E_h+i0))\inv = O(h\inv)$ as an operator from $L^{2,\a}(\R^n)$ to $L^{2,-\a}(\R^n)$ we get:

\begin{corollary}
$u_h = \bigo h 0 (h^{-\frac 12})$ in $L^{2,-\a}(\R^n)$. The same applies to $u_h^T$ for all $T \geq 0$.
\end{corollary}

\begin{proposition} \label{microloc-sh}
$S_h$ is microlocalized in ${N\G_0}$.
\end{proposition}
\begin{proof}
Let $q \in C_0^\infty(\R^{2n})$ supported outside ${N\G_0}$. We have:
\[
\begin{aligned}
\Opw (q) S_h(x)
& = \frac 1 {(2\pi h)^n} \int_{\G} \int_{\R^n}\int_{\R^n} e^{\frac ih \innp{x-y}\x} q(x,\x) A(z) S\left(\frac  {y-z}h\right) \, dy\,d\x\,d\s(z)\\
& = \frac 1 {(2\pi)^n} \int_{\G} \int_{\R^n}\int_{\R^n} e^{\frac ih \innp{x-z}\x}e^{-i\innp v \x} q(x,\x) A(z) S(v) \, dv\,d\x\,d\s(z)\\
\end{aligned}
\]
If $\partial_z \innp{x-z}\x=0$ and $\partial_\x\innp{x-z}\x=0$ then $x = z$ and $\x \in N_z\G$ so $A(z) q(x,\x) = 0$. According to the non-stationnary phase theorem, we have $\Opw(q) S_h = O(h^\infty)$ in $L^2(\R^n)$.
\end{proof}

\begin{proposition} \label{prop-loc-surf}
\begin{enumerate}[(i)]
\item Let $g\in \symbor$ equal to 1 in a neighborhood of $p\inv(\singl{E_0})$. We have:
\begin{equation} \label{surfa}
\nr{\Opw(1-g) (\hh-(E_h+i0))\inv }_{L^{2,\a}(\R^n) \to L^{2-\a}(\R^n)} = \bigo h 0 (1)
\end{equation}
\item Let $f \in \symbor$ equal to 1 in a neighborhood of $\bar {\negg_0}$, then in $L^{2,-\a}(\R^n)$:
\begin{equation}\label{surfb}
u_h = (\hh -(E_h+i0))\inv \Op(f) S_h + \bigo h 0 (\sqrt h)
\end{equation}
\item Moreover there exists $\tilde g \in C_0^\infty(\R)$ equal to 1 in a neighborhood of $E_0$ such that in $L^{2,-\a}(\R^n)$:
\begin{equation}\label{surfc}
(\hh -(E_h+i0))\inv \Op(1-f) S_h = (1-\tilde g)(\huh) (\hh -(E_h+i0))\inv \Op(1-f) S_h + \bigo h 0\big(h^{\frac 32}\big)
\end{equation}
\end{enumerate}
Similar results hold for $u_h^T$, $T \geq 0$.
\end{proposition}

\begin{proof}
\noindent
(i) 
For $\Im z > 0$ we have:
\[
\Op(1-g)(\hh-z)\inv = \Op(1-g)(\huh -z)\inv (1+hV_2(\hh-z)\inv ) 
\]
According to \cite{helfferr83} we have:
\[
(\huh -z)\inv = \Opw \left((p(x,\x)-z)\inv\right) + \bigo h 0 (h)
\]
Since $(p(x,\x)-z)\inv$ is bounded on $\supp(1-g)$ uniformly for $z$ close to $E_0$, $\Im z > 0$, the operator ${\Opw(1-g)(\huh -z)\inv}$ is uniformly bounded in $h>0$ and $z$ close to $E_0$, $\Im z >0$. Moreover ${(1+hV_2(\hh-z)\inv)}$ is uniformly bounded as an operator from $L^{2,\a}( \R^n)$ to $L^{2,-\a}( \R^n)$ so:
\begin{equation*}
\nr{\Opw(1-g) (\hh-z)\inv }_{L^{2,\a}(\R^n) \to L^{2-\a}(\R^n)} = \bigo h 0 (1)
\end{equation*}
uniformly in $z$. Taking the limit $z \to E_h + i0$ gives \eqref{surfa}.\\

\noindent
(ii)
Let $\Uc$ be a neighborhood of $\bar {\negg _0}$ in $\R^{2n}$ such that $f=1$ on $\Uc$. We can find $\e > 0$ such that $p\inv([E_0-2\e,E_0+2\e])\setminus \Uc$ does not intersect $\bar {N\G_0}$. Let $\h \in C^\infty_0(\R)$ supported in $]E_0-2\e,E_0+2\e[$ and equal to 1 on $]E_0-\e,E_0+\e[$. Since modulo $O(h^\infty)$ the operator $\h(H_1)$ is a pseudo-differential operator with symbol supported in $\supp (\h \circ p)$ and $S_h$ is microlocalized on $N\G_0$ we have in $L^{2,\a}(\R^n)$:
\[
(\hh -(E_h+i0))\inv \Op(1-f) \h(H_1) S_h = \bigo h 0 (h^\infty)
\]
On the other hand, as we proved \eqref{surfa} we see that:
\[
(\hh -(E_h+i0))\inv (1-\h)(\huh) = \bigo h 0 (1)
\]
so \eqref{surfb} follows since $\Opw(1-f)S_h = O(\sqrt h)$.\\

\noindent
(iii)
Let us refine this last estimate. Let $\tilde g \in C_0^\infty(\R)$ supported in $[E_0-\e,E_0+\e]$ and equal to 1 in a neighborhood of $E_0$. Since $(1-\h) \tilde g= 0$, we have:
\begin {eqnarray*}
\lefteqn{\tilde g (\huh) (\hh -z)\inv (1-\h)(\huh)}\\
&& = \tilde g (\huh) (\hh -z)\inv (1-\h)(\huh) (1-\tilde g)(\huh)\\
&& = \tilde g (\huh) (1+h(\hh-z)\inv V_2 )(\huh -z)\inv  (1-\h)(\huh)(1-\tilde g)(\huh)\\
&& = h \tilde g (\huh) \,  (\hh-z)\inv V_2 \, (1-\h)(\huh) \, (1-\tilde g) (\huh) \, (\huh -z)\inv
\end{eqnarray*}
It only remains to see that the operators  $(\hh-z)\inv V_2 (1-\h)(\huh)$ and $(1-\tilde g) (\huh)(\huh -z)\inv$ are bounded uniformly in $h \in ]0,1]$ and $z$ close to $E_0$ with $\Im z > 0$.
\end{proof}

As a first consequence of this proposition we see that the solution $u_h$ consentrates on $p\inv(\singl{E_0})$:

\begin{corollary} \label{cor-loc-surf}
If $q \in C_0^\infty(\R^n)$ has support outside $p\inv(\singl{E_0})$ then:
\[
\innp{\Opw(q) u_h} {u_h} \limt h 0 0
\]
\end{corollary}

\begin{proof}
Let $\tilde q \in C_0^\infty(\R^{2n})$ supported outside $p\inv(\singl{E_0})$ and equal to 1 on $\supp q$. We have:
\[
\innp{\Opw(q) u_h} {u_h} = \innp{\Opw(q) u_h} {\Opw(\tilde q) u_h} + \bigo h 0 (h^\infty) = \bigo h 0 (h)
\]
\end{proof}

\section{Around $\G$}

\subsection{WKB method}  \label{sec-bkw}

According to proposition IV.14 in \cite{robert} or lemma 10.10 in \cite{evansz} applied with the symbol $p_E : (x,\x) \mapsto \x^2 +V_1(x) - E_0$ we know that if $\t_0$ is small enough, then there exists a function $\f \in C^\infty([-3\t_0,3\t_0] \times \R^{2n})$ such that:
\begin{equation} \label{eqphi}
\left\{ \begin{array} l
\partial_t \f(t,x,\x) +  \abs{\partial_x \f(t,x,\x)}^2  + V_1(x) - E_0 = 0\\
\f(0,x,\x) = \innp x \x
        \end{array} \right.
\end{equation}
Moreover $\f$ is unique and:
\begin{equation} \label{defphi}
\begin{aligned}
\f(t,x,\x) 
& = \innp {\bar y(t,x,\x)} \x  + 2 \int _0^t \tilde\x(s,t,x,\x)^2\, ds - t p_E (x,\x)\\
& = \innp x \x - 2 \int_0^t \innp{\tilde\x(s,t,x,\x)} \x \, ds  + 2 \int _0^t \tilde\x(s,t,x,\x)^2\, ds - t p_E (x,\x)\\
& = \innp x \x - t p_E (x,\x) + t^2 r(t,x,\x)
\end{aligned}
\end{equation}
where $\bar y (t,x,\x)$ is the unique point in $\R^n$ such that $\bar x (t,\bar y(t,x,\x),\x) = x$ (note that $\bar y (t,x,\x)$ is well-defined for $t$ small enough, see \cite{robert}) and:
\begin{equation*}
\begin{aligned}
r(t,x,\x) 
& = \frac 2 {t^2}\int_{s=0}^t \int_{\t= s}^t \innp{\tilde \x (s,t,x,\x)}{ \nabla V_1 (\tilde x (\t,t,x,\x))} \, d\t\,ds  = \innp \x {\nabla V_1(x)} + \bigo t 0 (t)
\end{aligned}
\end{equation*}

\begin{proposition} \label{prop-bkw}
Let $f \in C_0^\infty(\R^{2n},\R)$. We can find a function $a(h) \in C_0^\infty([0,3\t_0] \times \R^{2n})$ such that:
\begin{equation}\label{cond-init}
a(0,x,\x,h) = f(x,\x) 
\end{equation}
and:
\begin{equation}\label{estim-bkw}
\sup_{t\in[0,3\t_0]} \nr{a(t,x,\x,h) e^{\frac ih \f(t,x,\x)} - e^{-\frac {it}h (\hh-E_h)} \left( f(x,\x) e^{\frac ih \innp x \x} \right)}_{L^2(\R^{2n})} \limt h 0 0
\end{equation}
\end{proposition}

\begin{proof}
We define:
\[
\y(s,t,x,\x) = \exp \left( \int_s^t ( i E_1 - V_2 (\tilde x (\t,t,x,\x) - \D_x \f(\t,\tilde x (\t,t,x,\x),\x))\, d\t \right)
\]
Then:
\begin{equation*}
a_0(t,x,\x) = f(\bar y (t,x,\x),\x) \y(0,t,x,\x)
\end{equation*}
and:
\begin{equation*}
a_1(t,y,\x) =   i \int_0^t \D_x a_{0}(s,\tilde x (s,t,x,\x), \x)  \y(s,t,x,\x) \, ds 
\end{equation*}
where for $0\leq s \leq t \leq \t_0$ we have set $\tilde x (s,t,x,\x) = \bar x (s , \bar y(t,x,\x),\x)$. Then we set $a(h) = a_0 + h a_1$. Initial condition \eqref{cond-init} is true and we can check that:
\begin{equation*} 
\left( \partial_t + 2 \partial_x \f . \partial_x +  \D_x \f + V_2 -i E_1 \right) a_0 (t,x,\x)=0
\end{equation*}
and:
\begin{equation*} 
\left( \partial_t + 2 \partial_x \f . \partial_x + \D_x \f + V_2 -i E_1 \right) a_1 (t,x,\x)=  i  \D_x a_{0} (t,x,\x)  
\end{equation*}
which, with \eqref{eqphi}, give \eqref{estim-bkw}. Note that the function $a(h)$ is of compact support and the absorption coefficient $V_2$ does not change the phase $\f$. Only $a$ depends on $V_2$ and the bigger $V_2$ is the faster $a$ decays with time.
\end{proof}

\begin{remark*}
If \eqref{h6} is replaced by: 
\begin{equation} \label{h6top}
E_h = \sum_{j=0}^N h^j E_j + O(h^{N+1}) \quad \text{for all } N \in \N
\end{equation}
then we can define:
\begin{equation*}
a_j(t,y,\x) =   i \int_0^t \left( \D_x a_{j-1}(s,\tilde x (s,t,x,\x), \x) + \sum_{k=0}^{j-2} E_{j-k} a_k(x,\tilde x (x,t,x,\x),\x) \right)  \y(s,t,x,\x) \, ds 
\end{equation*}
for all $j\geq 2$ and $a \sim \sum h^j a_j$ by Borel theorem (see \cite[th. 4.16]{evansz}). Then the rest is of size $O(h^\infty)$ instead of $o(1)$ in \eqref{estim-bkw} and hence in \eqref{B1J} and \eqref{b-zero} below.
\end{remark*}

\subsection{Critical points of the phase function}   \label{sec-ptcrit}

For $t\in[0,3\t_0]$, $x,\x \in \R^{n}$ and $z\in\G_1$ we write:
\begin{equation*} 
\p(t,x,z,\x) = \f(t,x,\x) - \innp z \x
\end{equation*}

In this section we study the critical points of $\p$ with report to $t,\x$ and $z$ with $t \in ]0,3\t_0]$, that is the solutions of the system:
\begin{equation} \label{syst-pt-crit}
\begin{cases} \partial_t \p(t,x,z,\x) = 0 \\ \partial_z \p(t,x,z,\x) = 0 \\ \partial_\x \p(t,x,z,\x) = 0  \\ t \in ]0,3\t_0]\end{cases}
\eqv
\begin{cases} \partial_t \f(t,x,\x) = 0 \\ \x \in N_z \G_1 \\ \partial_\x \f(t,x,\x) = z  \\t \in ]0,3\t_0] \end{cases}
\end{equation}

\begin{proposition} \label{prop-equiv-ptcrit}
Let $t \in ]0,3\t_0]$, $x,\x\in\R^n$ and $z \in \G$. If $(t,x,\x,z)$ is a solution of \eqref{syst-pt-crit} then $(z,\x) \in \negg_1$ and $x = \bar x(t,z,\x)$.
\end{proposition}

\begin{proof}
Assume that $(t,x,\x,z)$ is such a solution. We already know that $\x \in N_z \G_1$. By proposition IV.14 in \cite{robert} we have:
\begin{equation} \label{phi-et-flot}
(x,\partial_x \f(t,x,\x)) = \vf^t(\partial_\x \f (t,x,\x),\x)  = \vf^t(z,\x)
\end{equation}
and in particular: $x =\bar x(t,z,\x)$. Moreover, since $\f$ is a solution of \eqref{eqphi} we also have:
\[
p(z,\x) = p (x,\partial_x\f (t,x,\x)) = \abs{\partial_x \f (t,x,\x)}^2 + V_1(x) = E_0 - \partial_t \f(t,x,\x) = E_0
\]
which proves that $\abs \x^2 = E_0 - V_1(z)$.
\end{proof}

We prove that for $x$ close to $\G$ (but not on $\G_1$), there is a solution $(t,x,\x,z)$ for \eqref{syst-pt-crit}. By proposition \ref{prop-equiv-ptcrit}, this solution must be $(t_x,x,z_x, \x_x)$ (defined in proposition \ref{prop-superbij}), so we already have uniqueness.

We consider the function $\Phi$ defined as follows: 
for $y \in \tilde \G_1(0,3\t_0)$, $\x \in \R^n$, $\z \in T_{z_y}\G_1$ of norm less than 1, $\d \in [0,\g_1]$ (where $\g_1 \in]0,1]$ is chosen small enough for $\exp_ {z} (\d \z)$ being defined in $\G_2$ for all $z \in \G_1$ and $\z$ of norm less than 1) and $\theta \in ]0,3\t_0/\g_1]$ then:
\begin{equation} \label{def-grand-phi}
\Phi (\theta,y,\z,\x,\d) = \begin{cases} \frac 1 \d \left(\f(\d \theta , \bar x( \d t_y, z_y,\x_y) , \x ) - \innp{\exp_ {z_y} (\d \z)} \x\right) & \text{ if } \d\neq 0 \\ \innp{\x_y - \z} \x - \theta (\x^2 +V_1(z_y) -E_0) & \text{ if }\d=0\end{cases}
\end{equation}

For $\d \in ]0,\g_1]$, $t\in\left] 0,  \frac {3\t_0 \d} {\g_1}\right]$, $x \in \tilde \G_1(0,\d \t_0)$, $z$ such that $d_\G(z_x,z) \leq \d$ and $\x \in \R^n$ we have:
\[
\p(t,x,\x,z) = \d \Phi \left( \frac t \d , \bar x \left( \frac {t_x } \d ,z_x, \x_x\right) , \frac 1 \d (\exp_{z_x})\inv(z) , \x  ,\d  \right)
\]
Thus:
\begin{subequations} \label{equiv-phi}
\begin{eqnarray}
\partial_t \p (t,x,z,\x) = 0 &\eqv& \partial_\theta \Phi \left( \frac t \d ,  \bar x \left( \frac {t_x } \d ,z_x, \x_x\right), \frac 1 \d (\exp_{z_x})\inv(z), \x \right)=0\\
\partial_\x \p (t,x,z,\x) = 0 &\eqv& \partial_\x \Phi \left( \frac t \d , \bar x \left( \frac {t_x } \d ,z_x, \x_x\right), \frac 1 \d (\exp_{z_x})\inv(z), \x \right)=0\\
\partial_z \p (t,x,z,\x) = 0 &\eqv& \partial_\z\Phi \left( \frac t \d ,  \bar x \left( \frac {t_x } \d ,z_x, \x_x\right), \frac 1 \d (\exp_{z_x})\inv(z), \x  \right)=0
\end{eqnarray}
\end{subequations}

\begin{proposition}  \label{prop-gd-phi}
Let $K = \Tc\left( \left[\frac{\t_0}2,3\t_0\right] \times \bar \negg \right)$. There exists $\d_0 \in ]0,\g_1]$ such that for all $y  \in K$ and $\d\in[0,\d_0]$ the system:
\begin{equation} \label{eqptcrit}
\left\{ \begin{array} l 
	\partial_{\theta,\x,\z}\Phi ( \theta,y,\z, \x,\d) = 0 \\
        \theta \in \left] 0, \frac {3\t_0}{\g_1} \right]
        \end{array} \right.
\end{equation}
has a solution $(\theta,\x,\z) \in ]0,\t_0/\g_1] \times  \R^n  \times T_{z_y}\G$.
\end{proposition}

\begin{proof}
For $\d \in ]0,\g_1]$ we compute:
\begin{equation*} 
\begin{aligned}
\Phi(\theta,y,\z,\x,\d)
& = \frac 1 \d \left(\f (\d \theta , \bar x (\d t_y,z_y,\x_y),\x) - \innp{\exp_{z_y}(\d \z)} \x \right)\\
& = \frac 1 \d \Big( \innp{ \bar x (\d t_y,z_y,\x_y)} \x - \d \theta (\x^2 +V_1(\bar x (\d t_y, z_y,\x_y))-E_0) \\
& \hspace{1cm} + \d^2 \theta^2 r(\d\theta, \bar x (\d t_y,z_y,\x_y),\x) - \innp{\exp_{z_y}(\d \z)} \x \Big) \\
& = \innp{2 t_y \x_y- \z} \x -  \theta \left( \x^2 + V_1(z_y) - E_0\right) + \theta(V_1(\bar x (\d t_y,z_y,\x_y)) - V_1(z_y)) \\
& \hspace{1cm}+ \d \theta ^2 r(\d \theta ,\bar x (\d t_y,z_y,\x_y),\x,h)  - \frac 1 \d \innp{\exp_{z_y}(\d \z) - z_y - \d\z} \x\\
& = \innp{2 t_y \x_y-\z} \x -\theta \left(\x^2 + V_1(z_y) - E_0\right) + \d R(\theta,y,\x,\z,\d)
\end{aligned}
\end{equation*}
%
%
where $R$ is of class $C^1$. This proves that $\Phi$ is of class $C^1$. The point $(\theta,y,\z,\x,0)$ is a solution of \eqref{eqptcrit} if and only if:
\begin{equation*}
\left\{ \begin{array}{l}
\abs { \x} = \sqrt{E_0 - V_1(z_y)}\\
\x \in N_{z_y}^*\G\\
2 t_y \x_y - \z =  2 \theta \x\\
\theta  \in \left] 0, \frac {\t_0} {\g_1} \right]
\end{array} \right.
\end{equation*}

Let $y  \in K$. This system has a unique solution which we denote $(\tilde \theta _{y,0}, \tilde \z _{y,0},\tilde \x _{y,0})$. It is given by:
\begin{equation}
\tilde \theta _{y,0} =  {t_y}  ; \quad \tilde \z_{y,0} = 0; \quad \tilde \x _{y,0}   = \x_y
\end{equation}

For $z \in \G$ and $\x \in \R^n$ we denote by $\x_z^\sslash$ the orthogonal projection of $\x$ on $T_z\G$ and $\x_z^\bot = \x - \x_z^\sslash$. Then we have:
\begin{equation*} 
\begin{aligned}
\Hess_{\theta,\z,\x} \Phi (\theta,y, \z,\x,\d) = 
\begin{pmatrix}
0 & 0 & -2\trsp{\x_z^\sslash} & -2\trsp{\x_z^\bot} \\
0 & 0 & -I_d & 0 \\
-2 \x_z^\sslash & -I_d & -2 \th I_d &  0 \\
-2 \x_z^\bot & 0 & 0 & -2 \th I_{n-d}
\end{pmatrix}
+ \bigo \d 0 (\d)
\end{aligned}
\end{equation*}
and in particular:
\begin{equation*}
\det \Hess_{\theta,\z,\x} \Phi (\tilde\theta_{y,0},y,\tilde \z_{y,0},\tilde \x_{y,0},0)
= 2^{n-d+1} (-1)^{n-d}  t_y^{n-d-1} \abs{\x_z}^2
\end{equation*}

The derivative of the function:
\begin{equation*} 
(\theta , y,\z, \x ,  \d) \mapsto  \partial _{\theta,\z,\x} \Phi (\theta, y, \z, \x,\d) \in \R^{n+d+1}
\end{equation*}
with report to $\theta$, $\z$ and $\x$ at the point $(\tilde \th_{y,0},0,\tilde \z _{y,0} , \tilde \x_{y,0} ,0)$ is:
\begin{equation*}
\Hess_{\theta,\z,\x} \Phi((\tilde \th_{y,0},0,\tilde \z _{y,0} , \tilde \x_{y,0} ,0)) \in \GL_{n+d+1}(\R)
\end{equation*}
so we can apply the implicit function theorem around $(\tilde\theta_{y,0},y, \tilde \z_{y,0},\tilde \x_{y,0},0)$. We obtain that there exists $\d_y>0$, a neighborhood $\Vc_y$ of $y$ in $\R^n$ and a function $\f_{y}$ which maps $\Vc_y\times [0,\d_y]$ into a neighborhood $\Uc_y$ of $(\tilde \theta_{y,0},\tilde \z_{y,0},\tilde \x_{y,0})$ in $]0,\t_0/\g_1] \times T_{z_y}\G \times \R^n $ such that:
\begin{equation*}
\forall (v,\d) \in \Vc_y \times [0,\d_y], \forall (\theta,\z,\x) \in \Uc_y, \quad  \partial_{\theta,\z,\x} \Phi(\theta,v,\z,\x,\d) = 0 \eqv (\theta,\z,\x) = \f_y(v,\d)
\end{equation*}
$K$ is covered by a finite number of such neighborhoods $\Vc_y$. We get the result if we take for $\d_0$ the minimum of the corresponding $\d_y$.
\end{proof}

\begin{corollary}
For all $x \in\tilde \G(0,2\d_0\t_0)$ there is a unique $(t,z,\x) \in ]0,\t_0] \times \G \times \R^n$ such that $(t,x,z,\x)$ is a solution of the system \eqref{syst-pt-crit}. Moreover this solution is given by $(t_x,x,z_x,\x_x)$.
\end{corollary}

\begin{proof}
After proposition \ref{prop-equiv-ptcrit}, there only remains to prove existence. Let $x \in \tilde \G (0,2\d_0\t_0)$. There is $\d \in ]0,\d_0]$ such that $y = \bar x \left( \frac {t_x}\d , z_x ,\x_x \right) \in \tilde \G (\t_0,2\t_0)$. Proposition \ref{prop-gd-phi} and equations \eqref{equiv-phi} give the result.
\end{proof}

\subsection{Small times control} \label{sec-tps-ptt}

We can find a neighborhood $\Gc$ of $\negg_0$ such that for all $t \in [0,\t_0]$ and $(x,\x) \in \Gc$ we have $0<d_1 \leq \abs \x \leq d_2$ and $\bar x (t,x,\x) \in \tilde \G(2\t_0)$. We choose a function $\h \in C_0^\infty(\R)$ supported in $]-1,\t_0[$ and equal to 1 in a neighborhood of 0. For $f \in C_0^\infty(\R^{2n})$ supported in $\Gc$, we set:
\begin{equation} \label{eq3.11}
B_0(h) = \frac i h \int_0^\infty \h(t) e^{-\frac {it} h(\hh-E_h)} \Op(f) S_h \, dt
\end{equation}

Egorov theorem (see proposition \ref{prop-Egorov}) yields:
\begin{equation} \label{b0-loin}
\nr{\1{\R^n\setminus \tilde \G (2\t_0)} B_0(h)}_{L^2(\R^n)} = \bigo h 0 (h^\infty)
\end{equation}

\begin{proposition} \label{lem3.5}
If $\t_0 >0$ is small enough, then for all $\e > 0$, there exists $\t_1 \in ]0, \t_0]$ and $h_0 > 0$ such that for all $f \in C_0^\infty(\R^{2n})$ supported in $\Gc$ we have:
\begin{equation}
\forall h \in ]0,h_0],\quad  \nr{\1 {\tilde \G(\t_1)} B_0(h)}_{L^2(\R^n)}  \leq  \e 
\end{equation}
\end{proposition}

\begin{proof}
{\bf 1.}
If $\Four_h$ denotes the semiclassical Fourier transform we have:
\[
\begin{aligned}
\Four_h S_h(\x) 
& = h^{\frac{1-n-d}2} \int_{\R^n} \int_{\G} e^{-\frac ih \innp x \x} A(z) S\left(\frac {x-z}h \right) \, d\s(z) \, dx\\
& = h^{\frac {1+n-d} 2} \int_{\G} A(z) e^{-\frac ih \innp z \x}  \int_{\R^n}  e^{ -i \innp y \x} S(y)  \, dy \, d\s(z)\\
& = h^{\frac {1+n-d} 2} \hat S (\x) \int_{\G}  A(z)  e^{-\frac ih \innp z \x}  \, d\s(z)
\end{aligned}
\]
where $\hat S$ is the usual Fourier transform of $S$, and then:
\[
\begin{aligned}
\Op(f) S_h(x) 
& = \frac 1 {(2\pi h)^n} \int_{\R^n} e^{\frac ih \innp x \x} f(x,\x) \Four_h S_h(\x)\, d\x\\
& = \frac {h^{\frac {1+n-d} 2}}{(2\pi h)^n} \int_{\G} \int_{\R^n} e^{\frac ih \innp {x-z} \x} A(z) f(x,\x) \hat S(\x)\, d\x  \, d\s(z)  \\
\end{aligned}
\]
so:
\begin{equation} \label{expr-b1}
B_0 (h) = \frac{ih^{-\frac {1+n+d} 2}}{(2\pi)^{n}} \int_{0}^{+\infty} \!\!  \int_{\G}  \int_{\R^n}\h(t) A(z)  e^{-\frac ih \innp {z} \x} e^{-\frac{it}h(\hh-E_h)} \left(e^{\frac ih \innp {\cdot} \x}  f(\cdot,\x)\right) \hat S(\x)\, d\x \, d\s(z) \,dt
\end{equation}

Let $a$ and $\f$ given by WKB method (see section \ref{sec-bkw}). We define:
\[
J(x,h) = \int_{0}^\infty \!\! \int_{\G} \int_{\R^n}  \h(t)  e^{\frac ih (\f(t,x,\x) - \innp z \x)} a(t,x,\x,h) A(z) \hat S (\x) \, d\x\,d\s(z)\, dt
\]
so that by \eqref{estim-bkw}:
\begin{equation} \label{B1J}
B_0(h) = \frac{ih^{-\frac {1+n+d} 2}}{(2\pi)^{n}}  J(h) \left( 1 + \littleo h 0 (1)\right) \quad \text{in } L^2(\R^n)
\end{equation}

Let:
\[
\k(t,x,z,\x,h) = \h(t) a(t,x,\x,h) A(z) \hat S(\x)
\]
$\k$ is smooth and of compact support in $t,x,z,\x$ so all its derivatives are bounded. We recall that we wrote $\p(t,x,\x,z) = \f(t,x,\x) - \innp z \x$.\\

\noindent
{\bf 2.}
Let $N\in\N$. To estimate $J$, we define, for all $\d \in ]0,\d_0]$:
\[
J_\d(x) = \1 {\tilde \G(\d \t_0,2\d \t_0)}(x) \int_{\R} \int_{\G}\int_{\R^n}  e^{\frac ih \p(t,x,z,\x)} \k(t,x,z,\x,h)  \, d\x \, d\s(z)\, dt
\]

Let:
\[
J_\d^\sslash(x) = \1 {\tilde \G(\d \t_0,2\d \t_0)}(x) \int_{\R} \int_{\G}\int_{\abs{\x_z^\sslash}> d_1 \d }  e^{\frac ih \p(t,x,z,\x)} \k(t,x,z,\x,h)  \, d\x \, d\s(z)\, dt
\]
Since $\partial_z \p (t,x,z,\x)  = \x_z^\sslash$, $N$ partial integrations in $z$ show that:
\[
\abs{J_\d^\sslash(x)} \leq c\,\1{\tilde \G(\d \t_0,2\d \t_0)}(x) \left(\frac h \d\right)^N
\]
and hence:
\begin{equation} \label{estim-Jsslash}
\nr{J_\d^\sslash}_{L^2(\R^n)} \leq c \, h^N \, \d^{\frac {n-d}2 -N}
\end{equation}

\noindent
{\bf 3.}
By \eqref{defphi} we have:
\[
\partial_\x \p(t,x,z,\x) = x - (z + 2t\x) + t^2 \partial_\x r(t,x,\x)
\]
and hence:
\[
[x-(z+2t\x)]^\wedge.\partial_\x \p(t,x,z,\x) = \abs{x-(z+2t\x)} + t^2 [x-(z+2t\x)]^\wedge.\partial_\x r(t,x,\x)
\]
where $\hat x$ stands for $\frac x {\abs x }$.
For $t \leq \d \t_0 \min\left(1,\frac{\g_m}{4d_2}\right)$ ($\g_m$ is defined in proposition \ref{prop-superbij}) and $x \in \tilde\G (\d\t_0,2\d \t_0)$ we have:
\begin{equation*}
\begin{aligned}
\abs{x-(z+2t\x)}
& \geq \abs{x-z} -  2 t \abs \x \geq \d \t_0 \g_m  - 2td_2 \geq \frac {\d \t_0 \g_m}2  
\end{aligned}
\end{equation*}
and hence:
\begin{equation} \label{minor-tsmall}
\abs{x-(z+2t\x)} + t^2 [x-(z+2t\x)]^\wedge.\partial_\x r \geq \d \left( \frac {\t_0 \g_m}2 - M \t_0^2 \right)
\end{equation}
where $M = \nr{\partial_\x r}_{L^\infty([0,\t_0]\times \R^{2n})}$. Taking $\t_0$ smaller we may assume that the quantity in brackets is positive.

On the other hand if $t\in \left[ \d   \frac{2\t_0(2d_1 +  \g_M) + 1}{d_1} ,\t_0\right]$, $z_{xx}$ is a point of $\G_1$ for which $\abs{x-z_{xx}} = d(x,\G_1)$ and $\abs{\x_z^\sslash} \leq \d d_1$, then:
\[
\begin{aligned}
\abs{x-(z+2t\x)} 
& \geq \abs{z+2t\x - z_{xx}} - \abs {x- z_{xx}}\\
& \geq \abs{z+2t\x_z^\bot - z_{xx}} -2 \d \t_0 d_1 - 2 \d \t_0 \g_M \\
& \geq 2td_1 - 2  \d \t_0 ( 2d_1 + \g_M)
\end{aligned}
\]
since for $t$ small enough $d(z+2t\x_z^\bot,\G) = \abs{2t\x_z^\bot} \geq 2t\abs \x - 2t \abs{\x_z^\sslash}$. Thus:
\begin{equation}\label{minor-denom}
\abs{x-(z+2t\x)} + t^2 [x-(z+2t\x)]^\wedge.\partial_\x r \geq t(d_1-\t_0 M) + td_1 - 2 \d \t_0 (2d_1 + \g_M)\geq  \d + t \frac {d_1}2
\end{equation}
if $d_1 \geq 2 \t_0 M$, which may be assumed. In particular we have proved that there exists $C , c_0 >0$ such that:
\[
\forall \d \in ]0,\d_0], \forall t \in \left[ 0, \frac \d C\right] \cup [C\d,\t_0], \quad  \abs{x-(z+2t\x)} + t^2 [x-(z+2t\x)]^\wedge.\partial_\x r \geq c_0 \d
\]
on the support of $\1{\tilde \G(\d \t_0,2 \d \t_0)}(x) \k(t,x,z,\x,h) $. We get:
\[
\abs{\partial_\x^\a\frac{[x-(z+2t\x)]^\wedge}{\abs{x-(z+2t\x)} + t^2 [x-(z+2t\x)]^\wedge. \partial_\x r(t,x,\x)}} \leq {c_\a} {\d^{-\abs{\a}}}
\]
on this support, since the derivatives of $[x-(z+2t\x)]^\wedge$ with report to $\x$ are bounded for $t \in [0,\d/C]\cup [C \d ,\t_0]$ according to \eqref{minor-tsmall} and \eqref{minor-denom}. We choose a function $\h_1 \in C^\infty(\R)$ equal to 1 in a neighborhood of $\left] -\infty, \frac 1 {2C}\right] \cup [2C,+\infty]$ and zero on $\left[\frac 1 C,  C \right]$ and $\h_0 = 1 - \h_1$. Then we have $J_\d = J_\d^1 + J_\d^0+ J_\d^\sslash$ where, for $j\in\{0,1\}$:
\[
J_\d^j(x) = \1 {\tilde \G(\d \t_0,2\d \t_0)}(x)\int_{0}^\infty \int_{\G} \int_{\abs{\x_z^\sslash}\leq \d d_1}  \h_j\left(\frac t \d\right) e^{\frac ih \p(t,x,z,\x)} \k(t,x,z,\x,h)  \, d\x \, d\s(z) \, dt
\]

We consider the operator: 
\[
L : u \mapsto \left( (t,x,z,\x,h) \mapsto -ih \frac {[x-(z+2t\x)]^\wedge . \partial_\x u}{\abs{x-(z+2t\x)} + t^2 [x-(z+2t\x)]^\wedge.\partial_\x r}\right)
\]
The function $(t,x,z,\x,h)\mapsto \exp\left(\frac ih \p(t,x,z,\x) \right)$ is invariant by $L$ and the adjoint $L^*$ is given by:
\[
L^* : v \mapsto \left( (t,x,z,\x) \mapsto ih\divg_\x \left( \frac {[x-(z+2t\x)]^\wedge . v}{\abs{x-(z+2t\x)} + t^2 [x-(z+2t\x)]^\wedge.\partial_\x r}\right)\right)
\]
$N$ partial integrations with $L$ prove:
\[
\abs{J_\d^1(x)} \leq C_N \left(\frac h {\d}\right) ^N \1 {\tilde \G(\d \t_0,2\d \t_0)}(x)
\]
and hence:
\begin{equation} \label{estim-jd1}
\nr{J_\d^1} \leq C_N h^N \d^{\frac {n-d}2 - N}
\end{equation}

\noindent
{\bf 4.}
We now turn to $J_\d^0$. We recall that for all $z \in \G_1$ and $\z \in T_z\G_1$ of norm less than $\g_1$ then $\exp_z(\z)$ is well-defined (on $\G_2$) and $d_{\G_1} (z,\exp_z(\z)) = \abs{\z}$. For $\t_0$ small enough, if $x \in \tilde \G(\d \t_0, 2 \d \t_0)$ and $d_{\G_1}(z,z_x) \geq \g_1 \d$ then $\abs {x-z} \geq \frac {\g_1 \d} 2$ and $\abs {x-(z+2t\x)} \geq \frac{\g_1 \d} 4$. As a result we can do partial integrations with $L$ as before and see that modulo $O((h/\d)^N)$, $J_\d^0(x)$ is given by integration over $z$ in a neighborhood of radius $\d$ around $z_x$:
\begin{equation*}
\begin{aligned}
J_\d^0(x) 
& = \1 {\tilde \G(\d \t_0,2\d \t_0)}(x)\int_{0}^\infty \hspace{-0.2cm} \int_{ B_\G(z_x,\g_1 \d)} \int_{\abs{\x_z^\sslash}\leq \d d_1}  \h_0\left(\frac t \d\right) e^{\frac ih \p(t,x,z,\x)} \k(t,x,z,\x,h)  \, d\x d\s(z) \, dt\\
& \quad  + O \left( \left( h /\d\right)^N\right)
\end{aligned}
\end{equation*}
After the change of variables $t = \theta  \d$ and $z = \exp_{z_x}(\d \z)$, $\z \in T_{z_x}\G$, we get for $y\in\R^n$:
\begin{equation*}
\begin{aligned}
J_\d^0( \bar x (\d t_y, z_y,\x_y))
& =  \d ^{1+d} \1{\tilde \G(\t_0,2\t_0)} (y) \int \!\!\! \int_{\R^n} \int_{0}^{\infty}  \h_0 (\theta)  \tilde \k (\theta,y,\x,\z,h) e^{\frac ih \d \Phi ( \theta , y,\x,\z,\d)}  \, d\theta \, d\x  \, d\z  \\
& \quad  + O \left( \left( h /\d\right)^N\right)
\end{aligned}
\end{equation*}
where integral in $\z$ is over the ball or radius $\g_1$ in $T_{z_y}\G$ and: 
\begin{equation*}
\tilde \k (\theta,y,\x,\z,h,\d) =  \tilde \h (y)  \k (\d\theta,\d y, \x, \exp_{z_x}(\d\z), h) \partial_\z \exp_{z_x}(\d\z)
\end{equation*}
with $\tilde \h \in C_0^\infty(\R^{2n})$ supported in $\singl{\t_0/2\leq t_y \leq 3\t_0}$ and equal to 1 on $\singl{\t_0\leq t_y \leq 2\t_0}$. $\tilde \k(h,\d)$ is of compact support in $]0,+\infty[ \times (\R^{n} \setminus \G) \times  (\R^n \setminus \singl 0) \times T_{z_y}\G$. $\Phi$ is defined in \eqref{def-grand-phi}. For $y$ such that $\t_0/2\leq t_y \leq 3\t_0$ and $\d \in ]0,\d_0]$, there is by proposition \ref{prop-gd-phi} a unique $(\tilde \theta_{y,\d},\tilde \x_{y,\d},\tilde \z_{y,\d})$ such that $(\tilde \theta_{y,\d},y,\tilde \x_{y,\d},\tilde \z_{y,\d},\d)$ is a critical point of $\phi$ and $\tilde \theta > 0$. Moreover:
\begin{equation*}
\begin{aligned}
\partial_{\theta,\x,z} \Phi(\theta,y,z,\x,\d)
& = \Hess_{\theta,z,\x} \Phi ( \tilde \theta_{y,\d},y,\tilde \z_{y,\d},\tilde \x_{y,\d} ,\d) ((\theta,z,\x) - (\tilde \theta_{y,\d} ,\tilde \z _{y,\d}, \tilde \x_{y,\d})) \\
& \quad   + \bigo {(\theta,\z,\x)} {(\tilde \theta_{y,\d},\tilde\z_{y,\d}, \tilde \x_{y,\d})} ( \abspetit{ \theta - \tilde \theta_{y,\d}},\abspetit{\z - \tilde \z_{y,\d}}, \abspetit {\x - \tilde \x_{y,\d}})
\end{aligned}
\end{equation*}
and hence:
\begin{equation*}
\begin{aligned}
(\theta,\z,\x) - (\tilde \theta_{y,\d},\tilde \z_{y,\d} , \tilde \x_{y,\d}) 
& = \big[\Hess_{\theta,\z,\x} \Phi ( \tilde \theta_{y,\d},y,\tilde \z _{y,\d},\tilde \x_{y,\d},\d)\big] \inv     (\partial_{\theta,\z,\x} \Phi(\theta,\z,\x))  \\
& \quad   + \bigo {(\theta,\z,\x)} {(\tilde \theta_{y,\d},\tilde\z_{y,\d}, \tilde \x_{y,\d})} ( \abspetit{ \theta - \tilde \theta_{y,\d}},\abspetit{\z - \tilde \z_{y,\d}}, \abspetit {\x - \tilde \x_{y,\d}})
\end{aligned}
\end{equation*}
$y$ and $\d$ stay in a compact set and zero is never an eigenvalue of $\Hess_{\theta,\z,\x}\Phi (\tilde \theta_{y,\d},y, \tilde \z_{y,\d},\tilde \z_{y,\d},\d)$, so the norm of $\Hess_{\theta,\z,\x}(\theta,\z,\x) \inv $ is bounded.

As a consequence the quantity:
\begin{equation*}
\frac { \abs{ (\theta,\z,\x) - (\tilde \theta_{y,\d},\z_{y,\d}, \tilde \x_{y,\d})}}   {\abs{\partial_{\theta,\z,\x} \Phi (\theta,y ,z,\x,\d)}}
\end{equation*}
is uniformly bounded. So we can use theorems 7.7.5 and 7.7.6 in \cite{hormander1}, which give:
\begin{equation*}
\abs{J_\d^0(x) } \leq c \d^{1+d} \left( \frac h \d \right) ^{\frac {n+d+1} 2} \1 {\tilde \G(\d \t_0,2\d \t_0)}(x) + c \left( \frac h \d \right)^N \1 {\tilde \G(\d \t_0,2\d \t_0)}(x)
\end{equation*}
and thus:
\begin{equation} \label{jdzero}
\nr{J_\d^0} \leq c \d^{\frac {1} 2} h^{\frac {n+d+1} 2}+ c\, h^N \d ^{\frac {n-d}2-N}
\end{equation}

\noindent
{\bf 5.}
For $\g \in ]0,1]$ we define :
\[
\tilde J_\g(x) = \1 {\tilde \G(2\g \t_0)}(x) \int_{\R} \int_{\G}\int_{\R^n}  e^{\frac ih \p(t,x,z,\x)} \k(t,x,z,\x,h)  \, d\x \, d\s(z)\, dt
\]
$\tilde J_\g^\sslash$ is defined as $J_\d^\sslash$ with $ \1 {\tilde \G(\d \t_0, 2\d \t_0)}$ replaced by $ \1 {\tilde \G(2\g \t_0)}$. An estimate analog to \eqref{estim-Jsslash} holds for $\tilde J_\g^\sslash$. We now note $\h_+ = \1{[C,+\infty[}\h_1$, $\h_- = 1 - \h_+$, and:
\begin{equation*}
\tilde J_{\g}^{\pm}(x) = \1 {\tilde \G(2\g \t_0 )}(x) \int_\R \int_\G \int_{\abs{\x_z^\sslash} \leq \g d_1}  \h_\pm\left(\frac t \g \right) e^{\frac ih \p(t,x,z,\x)} \k(t,x,z,\x,h) \,d\s(z) \, d\x\, dt
\end{equation*}
As we did for $J_\d^1$ we see that:
\begin{equation} \label{estim-jd+}
\nr{\tilde J_\g^+} \leq C_N h^N \g^{\frac {n-d}2 - N}
\end{equation}

To estimate $J_\g^-$, we remark that we are integrating a bounded function over a set of size $O(\g)$ in $t$ and over $\{(z,\x), \big|\x_z^\sslash\big| \leq \g d_1\}$ whose volume is of size $O(\g^d)$, so:
\[
\abs{\tilde J_\g^-(x)}\leq c \g^{1+d} \1 {\tilde \G(2\g \t_0)}(x)
\]
Taking the $L^2(\R^n)$ norm in $x$ gives:
\begin{equation}\label{estim-jd-}
\nr{\tilde J_{\g}^-}_{L^2(\R^n)} \leq c \g^{1+\frac {n+d} 2}
\end{equation}

\noindent
{\bf 6.}
Estimates \eqref{estim-Jsslash}, \eqref{estim-jd1}, \eqref{jdzero}, \eqref{estim-jd+} and \eqref{estim-jd-} allow to conclude: let $\t_1\in]0,\d_0 \t_0]$ and $\m \in ]0, 1[$, we use a dyadic decomposition $\d = 2^{-m}$ with $h^{1-\m} < \d  < \t_1/\t_0$, that is ${\ln_2 (\t_0) -\ln_2(\t_1)} < m <  -(1-\m) \ln_2 h$. We write $m_- = \ln_2 (\t_0) -\ln_2(\t_1)$ and $m_+ =   -(1-\m) \ln_2 h$. Then:
\[
\nr{\1{\tilde \G(\t_1)} J} \leq \nr{\tilde J_{h^{1-\m}}} + \sum_{m_- < m < m_+} \nr{J_{2^{-m}}} 
\]
with:
\[
\begin{aligned}
\nr{\tilde J_{h^{1-\m}}}
& \leq \nr{\tilde J_{h^{1- \m}}^\sslash} + \nr{\tilde J_{h^{1- \m}}^-} + \nr{\tilde J_{h^{1- \m}}^+} \\
& \leq c_{N} \left(h^{(1-\m)\left(\frac {n+d} 2 + 1\right)} + h^{(1-\m)\frac {n-d} 2 + \m N} \right)\\
& \leq c_{N} h^{ \frac{n+d+1}2}  \left(  h^{\frac {1} 2-\m \left(\frac {n+d} 2 +1\right)} + h^{\m N - \frac 12 -d - \m \frac {n-d} 2} \right)
\end{aligned}
\]
and:
\begin{equation*}
\begin{aligned}
\sum_{m_- < m < m_+} \nr{J_{2^{-m}}}
& \leq  \sum_{m_-<m<m_+}  \left( \nr{J_{2^{-m}}^1} + \nr{J_{2^{-m}}^0}   + \nr{J_{2^{-m}}^\sslash} \right)\\
%
%
& \leq c_{N} \left( h^N\sum _{m\leq m_+}  \left( 2^{N-\frac{n-d}2}\right)^m +  h ^{\frac{n+d+1}2} \sum_{m_-\leq m }{2}^{-\frac m 2}   \right)\\
& \leq c_{N} \left(   h^{N-(1-\m)\left(N - \frac {n-d}2\right)}  +  h ^{\frac{n+d+1}2} \sqrt {\t_1}   \right)\\
& \leq c_{N} h^{ \frac{n+d+1}2}  \left(  h^{\m N - \frac 12 -d - \m \frac {n-d} 2}  + \sqrt {\t_1} \right)
\end{aligned}
\end{equation*}

We now take $\m > 0$ small enough to have $\nu : = \frac 12 - \m \left( \frac {n+d}2 + 1 \right) > 0$ and then $N$ big enough to have $\m N - \frac 12 - d -\m\frac{n-d}2 \geq 0$. This gives:
\begin{equation*}
\nr{\1{\tilde \G(\t_1)} J}_{L^2(\R^n)} \leq c\,h^{\frac{n+d+1}2} (\sqrt {\t_1} + h^\n)
\end{equation*}
If $\t_1$ and $h_0$ are small enough we have $c (\sqrt {\t_1} + h^\n) \leq \frac \e 2$ for all $h \in ]0,h_0]$. By (\ref{B1J}), if $h_0$ is small enough we finally reach the result:

\begin{equation*}
\nr{\1 {\tilde \G(\t_1)} B_0(h)}_{L^2(\R^n)} \leq    \e  
\end{equation*}
\end{proof}

For $z \in \G$ and $x \in \R^n$ we set:
\begin{equation} \label{def-tilde-psi}
\tilde \p _{x,z} : (t,\z,\x) \mapsto  \p (t, x , \exp_z (\z),\x)
\end{equation}
This is defined for $t \in ]0,\t_0]$, $\x \in \R^n$ and $\z$ in a neighborhood $\Uc_z$ of 0 in $T_z \G$.
Now for $ x \in \tilde \G (0,2\t_0)$ we let $\p(x) = \p(t_x,x,z_x,\x_x) = \f(t_x,x,\x_x) - \innp{z_x}{\x_x}$ and:
\begin{equation} \label{symb-b0}
b_0(x) = i (2\pi) ^{\frac {d+1-n}2} \frac{e^{\frac {i\pi}4 \sgn  \Hess \tilde \p_{x,z_x} (t_x,0,\x_x)}} {\abs{\det  \Hess \tilde \p_{z_x} (t_x,0,\x_x)}^{\frac 12} } A(z_x) a_0(t_x,x,\x_x) \hat S(\x_x) \h(t_x)
\end{equation}

\begin{proposition}  \label{prop-b0-moy}
Let $\Uc$ be a neighborhood of $\G_0$ in $\R^n$. Then on $\tilde \G(\t_0) \setminus \Uc$ the function $B_0$ is a lagrangian distribution of phase $\p$ and principal symbol $b_0$.
\end{proposition}

This means that $B_0$ is of the form $B_0(x) = e^{\frac ih \p(x)} b_0(x) + o(1)$. Note that if \eqref{h6top} holds we can have $B_0(x) = e^{\frac ih \p(x)} b(x,h)+O(h^\infty)$ where $b(x,h) \sim \sum h^j b_j(x)$ for some functions $b_j$, $j\geq 1$. See \cite{sogge} for more details about lagrangian distributions (in the microlocal setting).

\begin{proof}
Everything we need is already in the proof of proposition \ref{lem3.5}. By Egorov theorem there exists $\t_2 \in ]0, \t_0]$ such that:
\[
\1 {\tilde \G (\t_0) \setminus \Uc } B_0 = \1{\tilde \G (\t_2,\t_0)} B_0 + \bigo h 0 (h^\infty)
\]
 
Let us come back to the proof of \eqref{minor-tsmall} with $\d = \t_2$. We see that if $\bar \h \in C_0^\infty(\R_+^*)$ is such that $\bar \h(t) = \h(t)$ for $t \geq \frac {\g_m \t_2 \t_0}{4 d_2}$ then in $L^2(\tilde \G(\t_2,\t_0))$:
\[
\begin{aligned}
B_0(x)
= \frac{ih^{-\frac {1+n+d} 2}}{(2\pi)^{n}}  \int_{0}^\infty  \hspace{-0,2cm} \int_{\G} \int_{\R^n}  \bar \h(t)  e^{\frac ih \p(t,x,z,\x)} a(t,x,\x,h) A(z) \hat S (\x) \, d\x\,d\s(z)\, dt  \left( 1 + \littleo h 0 (1)\right)
\end{aligned}
\]

Moreover as we explained for $J_\d^0$ the only relevant part of integration on $z$ is around $z_x$, so:
\begin{equation} \label{b-zero}
\begin{aligned}
B_0(x,h)
& = \frac{ih^{-\frac {1+n+d} 2}}{(2\pi)^{n}}  \int_{0}^\infty \!\!\!\int_{\Uc_{z_x}}\!\! \int_{\R^n}  \bar \h(t)  e^{\frac ih \tilde \p_{x,z_x} (t,\z,\x)} a(t,x,\x,h) A(z) \hat S (\x) \Jac (\exp_{z_x})(\z) \, d\x\,d\z\, dt\\
& \quad  \times \left( 1 + \littleo h 0 (1)\right)
\end{aligned}
\end{equation}
Then, as we did to study $J_\d^0$, we use the results of section \ref{sec-ptcrit} and stationnary phase method to get the result (in particular the only stationnary point for $\tilde \p _{x,z_x}$ is $(t_x,0,\x_x)$.
\end{proof}

\begin{proposition} Let $x \in \tilde \G(\t_0)$. We have:
\begin{equation} \label{det-hess}
\abs{\det \Hess \tilde \p_{x,z_x}(t_x,0,\x_x) } = 2^{n-d+1} t_x^{n-d-1} \abs{\x_x}^2 + \bigo {t_x} 0 (t_x^{n-d})
\end{equation}
where the size of the rest is uniform in $x$.
\end{proposition}

\begin{proof}
(ii). By \eqref{defphi} we have:
\[
\begin{aligned}
\det \Hess \tilde \p_{x,z} (t,0,\x)
& = \begin{vmatrix} \partial_t^2 \f (t,x,\x) & 0 & - 2 \trsp {\x_z^\sslash}  & -2 \trsp{\x_z^\bot}  \\ 0 & A &-I_d & 0  \\  - 2  {\x_z^\sslash} & -I_d & -2 t I_d & 0  \\ -2 {\x_z^\bot} & 0 & 0 & -2t I_{n-d} \end{vmatrix} \left( 1 + \bigo t 0 (t) \right) \\
& = (-1)^{n-d} 2^{n-d+1} t^{n-d-1} \abs{\x_z^\bot}^2 +\bigo t 0 (t^{n-d})
\end{aligned}
\]
where for $1\leq i,j \leq d$:
\[
A_{ij} = -  \partial^2_{\z_i \z_j}  \innp{\exp_z (\z)}\x
\]
only appears in the rest, and $(\x_x)_{z_x}^\bot = \x_x$ since $(z_x,\x_x) \in \negg$.
\end{proof}

\section{Partial result for finite times}

\subsection{Intermediate times contribution}\label{sec-tps-inter}

We begin with a proposition which proves that for $w\in\R^{2n}$ and $q \in C_0^\infty(\R^{2n})$ supported close to $w$, then in the integral:
\[
u_h^T = \frac ih \int_0^{T} U_h^E(t) S_h\,dt
\]
only times around $t_{w,k}$ for $1 \leq k \leq K_w^T$ (and on a neighborhood of 0 if $w \in \negg$) give a relevant contribution.

\begin{proposition} \label{prop-terme-reste}
Let $w \in \R^{2n}$, $T > 0$ and $\tilde \h \in C_0^\infty(\R)$ a function which is zero near $t_{w,k}$ for $k\in\Ii 1 {K_w}$ (and 0 if $w \in \negg$). Then there exists a neighborhood $\Vc_{w,T}$ of $w$ in $\R^{2n}$ and a neighborhood $\Gc_{w,T} \subset \Gc$ of $\negg$ ($\Gc$ was defined in section \ref{sec-tps-ptt}) such that for all $q \in C_0^\infty(\R^{2n})$ supported in $\Vc_{w,T}$ and $f\in C_0^\infty(\R^{2n})$ supported in $\Gc_{w,T}$, we have in $L^2(\R^n)$:
\begin{equation*}
\Opw (q)\left(\frac ih \int_0^T \tilde \h (t)  U_h^E(t) \Op(f) S_h \, dt \right) = \bigo h 0 (h^\infty)
\end{equation*}
\end{proposition}

\begin{proof}
There exists a neighborhood $\Gc_{w,T} \subset \Gc$ of $\negg$ in $\R^{2n}$ and $\rho >0$ such that for all $\tilde w  \in \Gc$ and $t \in \supp \tilde \h$ we have:
\begin{equation*}
\abs {\vf^t(\tilde w) - w} \geq 2 \rho
\end{equation*}
Otherwise for all $m \in \N^*$ we can find $t_m \in \supp\tilde \h$ and $w_m \in \R^{2n}$ with $d(w_m),\negg) \leq \frac 1 m$ such that $\abs{\vf^{t_m}(w_m) - w} \leq \frac 1 m$. We can extract a subsequence so that $t_{m_k} \to t \in \supp \tilde \h$ and $w_{m_k} \to w_\infty \in\negg$. Then we have $\vf^t(w_\infty) = w$, which is impossible since $t \notin\{ t_{w,1},\dots,t_{w,K_w}\}$ ($\cup \singl 0$ if $w\in \negg$).

Let $\Vc_{w,T}$ be the ball $B(w,\rho)$ and $q \in C_0^\infty(\R^{2n})$ supported in $\Vc_{w,T}$. By Egorov theorem, we have for all $t\in [0,T]$:
\begin{equation*}
\nr{\Opw (q)  U_h^E(t) \Op(f)} = \bigo h 0 (h^\infty)
\end{equation*}
where the remainder is uniform in $t \in [0,T]$. An integration over $t$ gives the result.
\end{proof}

\begin{remark*}
Note that neither the neighborhoods $\Gc_{w,T}$ and $\Vc_{w,T}$ nor the size of the remainder can be uniform in $T$. That is the main reason why we cannot deal directly with $u_h$ and have to begin with a study of $u_h^T$.
\end{remark*}

Let $w \in \L$ and $\t_w = \min(t_{w,1},\t_0)$. We consider $\h_w \in C_0^\infty(\R)$ supported in $]0,2\t_w[$ and equal to 1 in a neighborhood of $\t_w$, and set:
\begin{equation*} 
B_w(h) = \frac ih \int_{t=0}^{\infty} \h_w(t) U_h^E(t) \Op (f) S_h \, dt
\end{equation*}
Moreover, for $k\in \Ii 1 {K_w}$ we denote:
\begin{equation} \label{def-bwk}
B_{w,k}(h) = \frac ih \int_{t=0}^{\infty} \h_w(t-t_{w,k}+\t_w) U_h^E(t)  \Op(f) S_h \, dt
\end{equation}

As in proposition \ref{prop-b0-moy} (and we do not even have to worry about very small times since $\h_w$ vanishes around 0) we see that $B_w(h)$ is a lagrangian distribution of submanifold
\begin{equation*} 
\begin{aligned}
\L_0 
& = \singl{(x,\partial_x \psi), x \in \tilde \G(0,2\t_0)}  = \singl{ \vf^{t_x} (z_x,\x_x), x\in\tilde \G(0,2\t_0)}\\
&  = \singl{\vf^t(z,\x), t \in ]0,2\t_0] , (z,\x) \in \negg}
\end{aligned}
\end{equation*}
and of principal symbol
\begin{equation*}
b_w(x) = i (2\pi) ^{\frac {d+1-n}2} \frac{e^{\frac {i\pi}4 \sgn  \Hess \tilde \p_{x,z_x} (t_x,0,\x_x)}} {\abs{\det  \Hess \tilde \p_{z_x} (t_x,0,\x_x)}^{\frac 12} } A(z_x) a_0(t_x,x,\x_x) \hat S(\x_x) \h_w(t_x)
\end{equation*}

\begin{proposition} \label{prop-b2k}
For all $w\in\L$ and $k \in \Ii 1 {K_w}$, $B_{w,k}(h)$ is a lagrangian distribution of lagrangian submanifold $\L_{w,k} := \vf^{t_{w,k}} \L_0$. We denote by $b_{w,k}$ and $\p_{w,k}$ the principal symbol and the phase of this distribution.
\end{proposition}

\begin{remark*}
Again, with \eqref{h6} this means that $B_{w,k}(h) = e^{\frac ih \p_{w,k}} b_{w,k} + o(1)$, but with assumption \eqref{h6top} we can write $B_{ w,k}(h) = e^{\frac ih \p_{w,k}} \tilde b_{w,k}(h) + O(h^\infty) $ where $\tilde b_{w,k}(h)$ is a classical symbol of principal symbol $b_{w,k}$.
\end{remark*}

\begin{proof}
We have:
\[
\begin{aligned}
B_{w,k}(h)
& = \frac ih \int_{t=0}^{\infty} \h_w(t-t_{w,k}+\t_w) U_h^E(t)  \Op(f) S_h \, dt\\
& = \frac ih \int_{t=-t_{w,k}+\t_w}^{\infty} \h_w(t)U_h^E(t+t_{w,k}-\t_w)   \Op(f) S_h \, dt \\
& = U_h^E(t_{w,k}-\t_w) B_w(h)
\end{aligned}
\]
It is known that $e^{-\frac {i(t_{w,k}-\t_w)} h (\huh-E_h)}$ turns a lagrangian distribution of submanifold $\L_0$ into a lagrangian distribution of submanifold $\vf^{t_{w,k}-\t_w} \L_0$ (see \cite{sogge,evansz}). We can similarly see that this also applies to $U_h^E(t_{w,k}-\t_w)$. Computations are actually close to what is done for WKB method, where we see that the imaginary part does not affect the phase factor but only the amplitude. Here again $V_2$ only appears in the symbol $b_{w,k}$ of the lagrangian distribution.
\end{proof}

We give another property of $B_{w,k}$ we are going to use in section \ref{sec-caract}:

\begin{proposition}  \label{prop-b2nul}
Let $w \in \L$. For all $k\in\Ii 1 {K_w}$ we have:
\begin{equation*}
(\hh-E_h) B_{w,k}(h) = 0 \quad \text{microlocally near $w$}
\end{equation*}
\end{proposition}

\begin{proof}
We have:
\begin{equation*}
\begin{aligned}
(\hh-E_h) B_{w,k}(h)
& = (\hh-E_h) \frac ih \int _0^{+\infty}  \h_w (t - t_k + \t_w) U_h^E(t) \Op(f) S_h\,dt  \\
& =  - \int _0^{+\infty}  \h_w (t - t_k + \t_w) \partial_t U_h^E(t) \Op(f) S_h\,dt\\
& =   \int _0^{+\infty}    \h'_w (t - t_k + \t_w)  U_h^E(t) \Op(f) S_h\,dt
\end{aligned}
\end{equation*}
As $\partial_t \h_w (t -t_k + \t_w)$ is zero near $t= t_j$ for $j\in\Ii 1 {K_w}$ (and $t=0$), the result is given by Egorov theorem as in the proof of theorem \ref{prop-terme-reste}.
\end{proof}

\subsection{Convergence toward a partial semiclassical measure}

We are now ready to give the semiclassical measure for $u_h^T$.

\begin{theorem} \label{prop3.6}
Let $T \geq 0$. There exists a nonnegative Radon measure $\m_T$ on $\R^{2n}$ such that for all $q\in C_0^\infty(\R^{2n})$ we have:
\begin{equation}  \label{conv-tot}
\innp{\Opw (q) u_h^T}{u_h^T} \limt h 0 \int q \, d\m_T
\end{equation}
\end{theorem}

\begin{proof}
\noindent{\bf 1. Localization around a point $w\in \R^{2n}$.}
We are going to show that for any $w\in \R^{2n}$ and $T \geq 0$, there is a neighborhood $\Vc_{w,T} \subset \R^{2n}$ such that for all $q \in C_0^\infty(\R^{2n})$ supported in $\Vc_{w,T}$ we have:
\begin{equation}  \label{conv-loc}
\innp{\Opw(q) u_h^T }{u_h^T} \limt h 0 \int q \, d\m_{w,T}
\end{equation}
where $\m_{w,T}$ is a Radon measure on  $\Vc_{w,T}$. If $w_1,w_2 \in \R^{2n}$ are such that $\Vc_{w_1,T} \cap \Vc_{w_2,T} \neq \emptyset$, then the two measures $\m_{w_1,T}$ and $\m_{w_2,T}$ coincide on $\Vc_{w_1,T} \cap \Vc_{w_2,T}$ (we only have to consider the two versions of \eqref{conv-loc} for $q \in C_0^\infty(\R^{2n})$ supported in $\Vc_{w_1,T} \cap \Vc_{w_2,T}$). Thus we can define the measure $\m_T$ on $\R^{2n}$ as the only measure which coincides with $\m_{w,T}$ on $\Vc_{w,T}$ for all $w \in \R^{2n}$. Then for all $q\in C_0^\infty(\R^{2n})$ a partition of unity and a finite numbers of applications of \eqref{conv-loc} give \eqref{conv-tot}.

So let $w \in \R^{2n}$. If $w \notin (\negg \cup \L)$ we can choose a neighborhood $\Vc_{w}$ of $w$ which does not intersect $ \negg \cup \L$. Proposition \ref{prop-terme-reste} with $\tilde \h = 1$ on $[0,T]$ shows:
\[
\innp{ \Opw(q) u_h^T }{u_h^T} \limt h 0 0
\]
for all $q \in C_0^\infty(\R^{2n})$ supported in $\Vc_w$. Hence we set $\m_{w,T}= 0$ on $\Vc_{w,T}$. This proves that if $\m_T$ exists then we must have:
\begin{equation} \label{mt-zero}
\m_T = 0 \quad \text {outside } \negg \cup \L
\end{equation}
We now assume that $w \in \negg \cup \L$.\\

\noindent {\bf 2. Localization around relevant times.}
%
Let $\d_w = 1$ if $w \in \negg$ and $\d_w = 0$ otherwise. We recall that $\h$ and $\h_w$ have been chosen in sections \ref{sec-tps-ptt} and \ref{sec-tps-inter}. By corollary \ref{dist-e0}, if $w\in \negg$ then $t_{w,1} \geq 3\t_0$ so for all $w \in \negg \cup \L$ supports of functions $\d_w \h$ and $\h_w(\cdot-t_{w,k}+\t_w)$ for $1\leq k \leq K_w^T$ are pairwise disjoint, so we can consider a function $\tilde \h \in C_0^\infty(\R,[0,1])$ such that:
\begin{equation*}
\forall t\in[0,T], \quad \d_w \h(t) + \sum_{k=1}^{K_w^T}  \h_w (t-t_k + \t_w) + \tilde \h(t) = 1
\end{equation*}

By proposition \ref{prop-terme-reste} there exists a function $f_{w,T} \in C_0^\infty(\R^{2n})$ equal to 1 around $\negg$ and a neighborhood $\Vc_{w,T}$ of $w$ in $\R^{2n}$ such that for $q$ supported in $\Vc_{w,T}$ we have in $L^2(\R^n)$:
\begin{equation*} 
\Opw(q)v_h^T = \Opw (q) \tilde u_h^T +\bigo h 0 (h^\infty)
\end{equation*}
where:
\[
v_h^T = \frac ih \int_0^T   U_h^E(t) \Op(f_{w,T}) S_h \, dt \quad \text{and}\quad \tilde u_h^T = \d_w B_{w,0}^T + \sum_{k=1}^{K_w^T} B_{w,k}^T
\]
with $B_{w,0}^T$ defined in \eqref{eq3.11} and the $B_{w,k}^T$ given by \eqref{def-bwk} with $f$ replaced by $f_{w,T}$. 
Let $\tilde g$ be given by proposition \ref{prop-loc-surf}. We have:
\begin{eqnarray}
\lefteqn{\innp{\Opw(q) \tilde u_h^T}{\tilde u_h^T}}   \label{first-dec}  \\
\nonumber && = \innp{\Opw(q) \left( v_h^T + (1-\tilde g)(\huh) (u_h^T - v_h^T) + O(h)\right)}{ v_h^T + (1-\tilde g)(\huh) (u_h^T - v_h^T) + O(h)}\\
\nonumber && = \innp{\Opw(q) v_h^T}{v_h^T} + \innp{\Opw(q)(u_h^T-v_h^T)} {(1-\tilde g)(\huh)v_h^T}\\
\nonumber && \quad  + \innp{\Opw(q)(1-\tilde g)(\huh)v_h^T} {u_h^T-v_h^T} + \bigo h 0 (\sqrt h)\\
\nonumber && = \innp{\Op(q) \tilde u_h^T}{\tilde u_h^T} + \bigo h 0 (\sqrt h)
\end{eqnarray}

\noindent{\bf 3. Definition of the measure $\m_{w,T}$.}
For $k\in\Ii 1 {K_w^T}$ and $\O$ a borelian set in $\Vc_{w,T}$ we define:
\begin{equation*}
\m_{w,T,k}(\O) = \int_{\R^n} \!\!  \1 \O (x, \partial \p_{w,k} (x)) \abs{b_{w,k}(x)}^2 \, dx  
\, ; \quad
\m_{w,T,0}(\O) = 
\d_w \int_{\R^n} \!\! \1 \O (x, \partial \p (x)) \abs{b_{0}(x)}^2 \, dx
\end{equation*}
and finally:
\begin{equation*}
\m_{w,T}(\O) = \sum_{k=0}^{K_w^T} \m_{w,T,k}
\end{equation*}
which defines a measure on $\Vc_{w,T}$. Note that all these measures are nonnegative. $\Vc_{w,T}$ and $\m_{w,T}$ are now fixed, and we have to prove that for any $\e >0$ and $q \in C_0^\infty(\R^{2n})$ supported in $\Vc_{w,T}$, there is $h_0 >0$ such that for all $h \in ]0,h_0]$:
\begin{equation} \label{contr-fin}
\abs{\innp{\Opw(q) u_h^T}{u_h^T} - \int q \, d\m_{w,T}} \leq \e
\end{equation}

Let $\e > 0$ and $q$ supported in $\Vc_{w,T}$. \eqref{first-dec} yields:
\begin{equation} \label{contr-decomp}
\abs{ \innp{\Opw (q)  u_h^T}{ u_h^T} - \innp{\Opw (q) \tilde u_h^T }{\tilde u_h^T}} \leq \frac \e 9
\end{equation}
with $h \in ]0,h_0]$ for some $h_0 > 0$.\\

\noindent{\bf 4. Self-intersections of $\L$.}
Let $j,k\in \Ii 1 {K_w}$ with $j\neq k$ ($j,k\in \Ii 0 {K_w}$ if $w \in \negg$). $\L_{w,j}\cap \L_{w,k}$ is a closed set of measure 0 in the smooth manifold $\L_{w,j}$, hence by regularity of the measure on $\L_{w,j}$, for all $m\in\N$  we can find an open subset $U_j^m$ of $\L_{w,j}$ of measure less than $\frac 1 m$ such that $\L_{w,j} \cap \L_{w,k} \subset U_j^m$. We can find an open sett $V_j^m$ in $\R^{2n}$ of measure less than $\frac 1 m$ such that $U_j^m = V_j^m \cap \L_{w,j}$, and by Uryshon lemma there exists a function $\g_j^m \in C_0^\infty(\R^{2n},[0,1])$ equal to 1 outside $V_j^m$ and zero in a neighborhood of $\L_{w,j}\cap \L_{w,k}$. We construct similarly a function $\g_k^m$ interverting $j$ and $k$, we set $\g_{j,k}^m = \g_j^m \g_k^m$ and finally:
\begin{equation}
\g_{m} = \prod_{1\leq j<k\leq K_w^T} \g_{j,k}^m  \hspace{1cm} \left( \text{ or }\prod_{0\leq j<k\leq K_w^T} \g_{j,k}^m \text{ if } w\in\negg\right)
\end{equation}
so that the sets $\L_{w,k}\cap \Vc_{w,T}$ for $1\leq k \leq K_w^T$ (or $0\leq k \leq K_w^T$) do not intersect on the support of $\g_{m}$ and:
\begin{equation}
\operatorname{mes}_\L \left(\supp(1 - \g_{m}) \cap  \left(\cup_{j=0}^{K_w^T}\L_{w,k} \right)\right) \leq \frac 1 m
\end{equation}

For all $k \in \Ii 0 {K_w^T}$, the support of the function $x \mapsto (1-\g_m)(x,\partial \p_k (x))$ is of measure less than $\frac C m$ in $\R^n$ where $C$ only depends on $\G$.
$\Opw (\g_m) B^T_{w,k}$ is a lagrangian distribution microlocally supported in $\L_{w,k} \cap \supp (\g_m)$ with symbols uniformly bounded in $h$ and $k$, so there is $c \geq 0$ such that for all $h \in ]0,h_0]$:
\begin{equation}  \label{contr-recoup}
\abs{ \tilde u_h^T- \Opw(\g_m) \tilde u_h^T } \leq \frac {c} {m}
\end{equation}

Moreover, for $j\neq k \in \Ii {0} {K_w}$ the distributions $\Opw(q  \g_m) B^T_{w,j}$ and $\Opw(\tilde q \g_m) B^T_{w,k}$ have disjoint microsupports, so we have:
\begin{equation} \label{contr-crois}
\innp{\Opw(q   \g_m) B^T_{w,j}} {\Opw(\tilde q \g_m) B^T_{w,k}} = \bigo h 0 (h^\infty)
\end{equation}
Taking $m \in \N$ large enough and using \eqref{contr-decomp}, \eqref{contr-recoup} et \eqref{contr-crois}, we obtain for all $h \in ]0,h_0]$:
\begin{equation}  \label{decom-total}
\abs{  \innp{\Opw(q) u_h}{u_h} - \d_w \innp{\Opw(q\g_m) B_{w,0}^T}{B_{w,0}^T} -  \sum_{k=1}^{K_w^T} \innp{\Opw(q \g_m) B^T_{w,k}}{B^T_{w,k}}} \leq \frac \e 3 
\end{equation}

\noindent{\bf 5. Convergence for intermediate times.}

Let $k \in \Ii 1 {K_w^T}$. We know that $B^T_{w,k}$ is a lagrangian distribution of phase $\p_{w,k}$ and of principal symbol $b_{w,k}$, hence we have:
\begin{equation*} 
\innp{\Opw (q) \Opw (\g_m) B^T_{w,k}}{B^T_{k,w}} = \int_{\R^n} q(x,\partial \p_{w,k}(x)) \g_m(x,\partial \p_{w,k} (x)) \abs{b_{w,k}(x)} ^2 \, dx  +  \littleo h 0 (1)
\end{equation*}
If $m$ is large enough and $h_0$ small enough, we have for all $h \in ]0,h_0]$:
\begin{equation} \label{conv-inter}
\abs{ \innp{\Opw (q) \Opw (\g_m) B^T_{w,k}}{B^T_{w,k}} - \int_{\R^n} q(x,\partial \p_{w,k} (x)) \abs{b_{w,k}(x)} ^2 \, dx} \leq \frac \e {3K_w^T}
\end{equation}

\noindent{\bf 6. Convergence for small times.}

It only remains to consider the term $\d_w \innp{\Opw (q) \Opw (\g_m) B^T_{w,0}}{B^T_{w,0}}$. We assume that $w$ belongs to $\negg$.

Let $\t_1 \in ]0,\t_0]$ and $v\in C^\infty_0(\R^{2n},[0,1])$ such that $\supp v \subset \tilde \G (\t_1)$ and $v$ is equal to 1 in a neighborhood of $\supp A$. By proposition \ref{lem3.5}, if $\t_1>0$ is small enough we have:
\begin{equation}
\nr{v B^T_{w,0}}_{L^2(\R^n)} \leq \frac \e 6
\end{equation}

On the other hand, since $(1-v)$ vanishes around $\supp A$, we can write $(1-v(x))B_{w,0}^T$ as a lagrangian distribution (see proposition \ref{prop-b0-moy}):
\begin{eqnarray*}
\lefteqn{\innp{\Opw(q)\Opw(1-v) \Opw(\g_m) B^T_{w,0}}{B^T_{w,0}}}\\
&& = \int_{\R^n} (q \g_m)(x, \partial \p(x)) (1-v(x)) \abs{b_0(x)}^2 \, dx + \littleo h 0 (1)
\end{eqnarray*}
Thus, if $\t_1$ and $h_0$ are small enough, then for all $h \in ]0,h_0]$:
\begin{equation} \label{conv-petit}
\abs{\innp{\Opw(q)\Opw(1-v) \Opw(\g_m) B^T_{w,0}}{B^T_{w,0}} - \int q \, d\m_{w,0}} \leq \frac \e 6
\end{equation}

\noindent {\bf 7. Conclusion.}
According to \eqref{decom-total}, \eqref{conv-inter} and \eqref{conv-petit}, we can conclude that \eqref{contr-fin} holds.
\end{proof}

\section{Convergence toward a semiclassical measure}

\subsection{Large times control}\label{sec-tps-gd}

For $R > 0$, $d> 0$ and $\s \in ]-1,1[$ we note:
\begin{eqnarray*}
\G_\pm(R,d,\s)&  =& \singl{(x,\x)\in\R^{2n} \tqe \abs x \geq R,   \abs \x  \geq d  \text{ and } \innp x \x \gtrless \s \abs x \abs \x}
\\
\G_\pm(d,\s)&  =& \singl{(x,\x)\in\R^{2n} \tqe \abs \x \geq  d   \text{ and } \innp x \x \gtrless \s \abs x \abs \x}
\end{eqnarray*}

As mentionned in the introduction, the following proposition states that the outgoing solution $u_h$ is microlocally zero in the incoming region. The proof of this proposition is postponed to section \ref{sec-incoming}.

\begin{proposition} \label{prop-incoming}
Let $d>0$, $\s \in ]0,1[$ and $E_h$ such that $\Im E_h >0$ or $E_h$ is positive and satisfies \eqref{hyp2}. Then there exists $R >0$ such that if $\o_-,\o \in \Sc_0$ are supported in $\G_-(R,d,-\s)$ (respectively outside $\G_-(R_1,d_1,-\s_1)$ for some $R_1<R$, $d_1<d$ and $\s_1<\s$) then:
\[
\nr{\Op(\o_-) (\hh -(E_0+i0))\inv \Op(\o)} = \bigo h 0 (h^\infty)
\]
\end{proposition}

We now use this proposition to show that for $T$ large enough, $\innp{\Opw(q) u_h^T}{u_h^T}$ is a good approximation of $\innp{\Opw(q) u_h}{u_h}$.

\begin{proposition}  \label{tps-grands}
Let $q\in C_0^\infty(\R^{2n})$ be supported in $p\inv(I)$ and $\e > 0$. Then there exists $T_0 \geq 0$ such that for all $T \geq T_0$ we can find $h_T > 0$ which satisfies:
\[
\forall h \in ]0,h_T], \quad  \abs{\innp{\Opw(q)u_h}{u_h} - \innp{\Opw(q)u_h^T}{u_h^T}} \leq \e
\]
\end{proposition}

\begin{proof}
\noindent {\bf 1.} Let $R_b \geq 0$ such that $\G \subset B_{\R^n}(R_b)$, $\supp q \subset B_x(R_b) = \singl{(x,\x) \in \R^{2n} \tqe \abs x < R_b}$ and any trajectory of energy in $J$ which leaves $B_x(R_b)$ never comes back (and goes to infinity). Let $\h \in C_0^\infty(\R^n)$ supported in $B(2 R_b)$ and equal to 1 on $ {B} (R_b)$. Let $Q \in C_0^\infty(\R^{2n})$ supported in $p\inv(J)$ and equal to 1 in a neighborhood of $p \inv (I) \cap  {B_x} (2 R_b)$ and of $\supp q$. Let $T\geq 0$ and $\o_-$ equal to 1 in the incoming region $\G_-(R_b,-1/2)$ and zero outside $\G_-(R_b/2,-1/4)$. We have:
\begin{equation} \label{decomp-uh}
\begin{aligned}
\Opw(Q) u_h
& = \frac i h \int_{t=0}^T \Opw(Q) U_h^E(t) S_h \, dt + \Opw(Q) U_h^E(T) u_h\\
& = \Opw(Q) u_h^T + \Opw(Q)  U_h^E(T) \Opw(Q) u_h \\
& \quad + \Opw(Q) U_h^E(T) \Opw(1-Q) \h(x)   u_h\\
& \quad + \Opw(Q) U_h^E(T) \Opw(1-Q) (1-\h(x)) \Op (\o_-) u_h \\
& \quad + \Opw(Q) U_h^E(T) \Opw(1-Q) (1-\h(x)) \Op (1-\o_-) u_h\\ 
\end{aligned}
\end{equation}
For $T$ large enough the last three terms are $O_{h\to 0}(\sqrt h)$ respectively by the localization close to the $E_0$-energy hypersurface (proposition \ref{prop-loc-surf}, which implies that $\Opw(1-Q) \h(x) u_h$ is small), estimates on the incoming region ($\Opw(\o_-) u_h$ is small by proposition \ref{prop-incoming}, changing quantization is harmless here) and Egorov theorem ($\Opw(Q) U_h^E(T) \Op(1-\o_-)  (1-\h(x))$ is small). Hence we have:
\begin{equation} \label{eg-uT}
 \left(1-\Opw(Q) U_h^E(T) \Opw(\tilde Q) \right) \Opw(Q) u_h =  \Opw(Q) u_h^T + \bigo h 0 (\sqrt h)
\end{equation}
where $\tilde Q \in C_0^\infty(\R^{2n})$ is supported in $p\inv(J)$ and equal to 1 on the support of $Q$. Furthermore:
\begin{equation*}
\nr{ \Opw(Q) u_h^T}^2 =  \innp{\Opw(Q)^2 u_h^T} {u_h^T} \limt h 0 \int Q^2 d\m_T < + \infty
\end{equation*}
Hence for any (large enough) fixed $T$, the right-hand side of \eqref{eg-uT} is uniformly bounded in $h$. Moreover, by proposition \ref{prop-super-egorov}, there exists $T_0$ such that for all $T \geq T_0$ there is $h_T >0$ which satisfies:
\begin{equation*}
\forall h \in ]0,h_T], \quad \nr{\Opw(Q) U_h^E(T) \Opw(\tilde Q)} \leq \frac 12
\end{equation*}
As a consequence, the operator $(1-\Opw(Q) U_h^E(T) \Opw(\tilde Q))$ is invertible and its inverse is bounded uniformly in $T \geq T_0$ and $h\in]0,h_T]$. This proves that the quantity: 
\begin{equation*} 
\Opw(Q)  u_h =   \left(1-\Opw(Q) U_h^E(T) \Opw(\tilde Q)\right)\inv \Opw(Q) u_h^T + \bigo h 0 (\sqrt h)
\end{equation*}
is bounded uniformly in $h \in ]0,h_T]$ for fixed $T \geq T_0$ and hence is bounded uniformly for $h$ small enough since the left hand side does not depend on $T$.\\

\noindent
{\bf 2.}
As for \eqref{decomp-uh} we see that:
\begin{equation} \label{decomp-uh-2}
\begin{aligned}
  \Opw(q) u_h
& = \Opw(q) u_h^T + \Opw(q)  U_h^E(T) \Opw(Q) u_h \\
& \quad + \Opw(q) U_h^E(T) \Opw(1-Q) \h(x)   u_h\\
& \quad + \Opw(q) U_h^E(T) \Opw(1-Q) (1-\h(x)) \Op (\o_-) u_h \\
& \quad + \Opw(q) U_h^E(T) \Opw(1-Q) (1-\h(x)) \Op (1-\o_-) u_h\\ 
\end{aligned}
\end{equation}
And as for \eqref{decomp-uh} the last three terms are $\bigo h 0 (\sqrt h)$ by localization close to $E_0$-energy hypersurface, estimates in the incoming region and Egorov theorem. Moreover the second term is:
\[
\Opw(q)  U_h^E(T) \Opw(Q) u_h = \Opw(q)  U_h^E(T) \Opw(\tilde Q) \left(\Opw(Q) u_h \right) + \bigo h 0 (\sqrt h)
\]
But $\Opw(Q) u_h $ is bounded uniformly in $h$ and the operator $\Opw(q)  U_h(T) \Op(\tilde Q)$ is of norm less than any $\d>0$ for $T$ big enough and $h$ small enough (depending of the chosen $T$). Hence we have proved:
\begin{equation} \label{controle-temps-grands}
\forall \d > 0, \exists T_0 \geq 0, \forall T \geq T_0, \exists h_T > 0, \forall h \in ]0,h_T], \quad \nr{ \Opw (q) (u_h - u_h^T)} \leq \d
\end{equation}
and in particular:
\begin{equation}\label{estim-qut}
\exists C \geq 0, \forall T \geq T_0,  \forall h \in ]0,h_T], \quad \nr{\Opw (q) u_h^T} \leq C
\end{equation}

We consider $\tilde q \in C_0^\infty(\R^{2n})$ supported in $p\inv(I)$, equal to 1 on $\supp q$ and such that $Q= 1$ on a neighborhood of $\supp \tilde q$. We can assume that \eqref{controle-temps-grands}-\eqref{estim-qut} hold for $q$ and $\tilde q$. Let $\d \in \left]0, \frac \e {4C}\right]$ and then $T$ and $h_T$ given by \eqref{controle-temps-grands}. For all $h \in ]0, h_T]$ we have:
\begin{eqnarray*}
\lefteqn{ \abs{ \innp{\Opw(q) u_h}{u_h} -  \innp{\Opw(q) u_h^T}{u_h^T}}}\\
&& = \abs{ \innp{\Opw(q) u_h}{\Opw(\tilde q)u_h} -  \innp{\Opw(q) u_h^T}{\Opw(\tilde q)u_h^T}} + \bigo h 0 (h^\infty) \\
&& \leq  \abs{ \innp{\Opw(q) (u_h - u_h^T)}{\Opw(\tilde q)u_h^T}} +  \abs{ \innp{\Opw(q) u_h}{\Opw(\tilde q)(u_h - u_h^T) }}+ \bigo h 0 (h^\infty)\\
&& \leq  \d \left(\nr{\Opw(q) u_h^T} + \nr{\Opw(\tilde q) u_h^T} \right) + \bigo h 0 (h^\infty)\\
&& \leq \frac \e 2 + \bigo h 0 (\sqrt h)
\end{eqnarray*}
and this last quantity is less than $\e$ if we choose $h$ small enough.
\end{proof}

\subsection{Convergence of the partial semiclassical measure}

\begin{proposition}
There exists a Radon measure $\m$ on $\R^{2n}$ such that for all $q \in C_0^\infty(\R^{2n})$:
\[
\int q \, d\m_T \limt T {+\infty} \int q \, d\m
\]
and we have:
\[
\innp{\Opw(q) u_h}{u_h} \limt h 0 \int q \, d\m
\]
\end{proposition}

\begin{proof}
{\bf 1.}
We can assume that for any $w \in \R^{2n}$, the family of neighborhoods $\Vc_{w,T}, T \geq 0,$ decreases when $T$ increases. Let $T_1\leq T_2 \in \R_+$. For $w \in \R^{2n}$ and $q \in C_0^\infty(\R^{2n})$ supported in $\Vc_{w,T_2} \subset \Vc_{w,T_1}$ we have:
\[
\int q \, d\m_{T_1} = \int q \, d\m_{w,T_1} = \sum_{k=0}^{K_w^{T_1}}\int_q \, d\m_{w,T_1,k} \leq \sum_{k=0}^{K_w^{T_2}}\int  q \, d\m_{w,T_2,k} = \int q \, d\m_{T_2}
\]
Since any $q \in C_0^\infty(\R^{2n})$ can be written as a finite sum $\sum q_i$ where $q_i$ is supported in $\Vc_{w_i,T_2}$ for some $w_i$, the same applies for all $q \in C_0^\infty(\R^{2n})$. This proves that $\int q \, d\m_T$ grows with $T$, and hence has a limit in $\R_+ \cup \{ + \infty\}$ when $T$ goes to $+\infty$.\\

\noindent
{\bf 2.}
If $\supp q \cap p\inv(\singl{E_0}) = \emptyset$, then 
\[
\int q \, d\m_T = 0 \limt T {+\infty} 0
\]
This is consistent with corollary \ref{cor-loc-surf}.\\

\noindent
{\bf 3.}
Now let $q\in C_0^\infty(\R^{2n})$ supported in $p\inv(I)$, $\tilde q$ and $C$ as in the proof of proposition \ref{tps-grands} (see \eqref{estim-qut}). We have:
\[
\int q \, d\m_T = \lim_{h\to 0} \innp{\Opw(q) u_h^T}{u_h^T} = \lim_{h\to 0} \innp{\Opw(q) u_h^T}{\Opw(\tilde q) u_h^T} \leq C^2
\]

As a result, $\int q \, d\m_T$ as a finite limit when $T$ goes to $+\infty$. This limit defines a nonnegative (each $\m_T$ is a nonnegative measure) linear form on $C_0^\infty(\R^{2n})$. Let $K$ be compact in $\R^{2n}$ and $Q \in C_0^\infty(\R^{2n})$ equal to 1 on $K$. Then for all $q \in C_0^\infty(\R^{2n})$ supported in $K$ we have:
\[
\abs{\int q \, d\m}  \leq  \lim_{T \to \infty} \int \abs q \, d\m_T \leq  \nr q _\infty \lim_{T \to \infty} \int Q \, d\m_T \leq c \nr q _\infty
\] 
and hence this limit is a continuous function of $q$ (is the space of compactly supported continuous functions). Thus the application $q \mapsto \lim_{T\to+\infty} \int q\, d\m_T$ can be extended to a nonnegative continuous linear form on the space of compactly supported continuous functions so, by Riesz theorem, there is a nonnegative Radon measure $\m$ on $\R^{2n}$ such that:
\[
\lim_{T \to \infty} \int q \, d\m_T = \int q \, d\m
\]

\noindent
{\bf 4.}
For $q \in C_0^\infty(\R^{2n},[0,1])$ there exists $T\geq 0$ such that:
\[
0 \leq \int q \, d\m - \int q \, d\m_T \leq \frac \e 3
\]
According to proposition \ref{tps-grands}, if $T$ is chosen large enough there is $h_T > 0$ such that:
\[
\forall h \in ]0,h_T], \quad \abs{ \innp{\Opw(q) u_h}{u_h} -  \innp{\Opw(q) u_h^T}{u_h^T}} \leq \frac \e 3
\]
and by theorem \ref{prop3.6}, there is $h_0 \in ]0,h_T]$ such that for all $h \in ]0,h_0]$ we have:
\[
\abs{ \innp{\Opw(q) u_h^T}{u_h^T} - \int q \, d\m_T} \leq \frac \e 3
\]

Hence we get:
\[
\forall h \in ]0,h_0], \quad \abs{ \innp{\Opw(q) u_h}{u_h} - \int q \, d\m} \leq \e
\]
which proves the proposition.
\end{proof}

\subsection{Characterization of the semiclassical measure}   \label{sec-caract}

We now finish the proof of theorem \ref{th2.1}:

\begin{proof}
\noindent {\bf 1.}
Statement (i) is already proved and similarly, (ii) is a consequence of the estimate in the incoming region (see proposition \ref{prop-incoming}).\\

\noindent {\bf 2.} 
Let $q \in C_0^\infty(\R^{2n})$ such that $\supp q \cap (\negg \cup \L) = \emptyset$. We have:
\[
\begin{aligned}
\int q \, (H_p + 2\Im E_1 + 2V_2) d\m
& = \int(-H_p + 2\Im E_1 + 2V_2)  q \, d\m \\
& = \lim_{T \to \infty} \int(-H_p + 2\Im E_1 + 2V_2)  q \, d\m_T \\
& =0
\end{aligned}
\]
according to \eqref{mt-zero} since the support of $(-H_p + 2\Im E_1 + 2V_2)  q$ does not meet $\negg \cup \L$.\\

\noindent {\bf 3.} 
Let $w \in \L$, $T \geq 0$ and $q \in C_0^\infty(\R^{2n})$ such that $\supp q \subset \Vc_{w,T}$.
 
Since $2ih\Im E_1 = E_h - \bar E_h + \littleo h 0 (h)$ and $H_{p}(q) = \{ p,q\}$ is the principal symbol of the operator $\frac ih [\huh, \Opw(q)]$, we have: 
\[
\Opw(H_p(q)) = \frac ih [\huh,\Opw(q)] + h \Opw (r_1) + \bigo h 0 (h^2)
\]
for some symbol $r_1 \in C_0^\infty(\R^{2n})$. But $\innp{\Opw (r_1) B_{w,k}^T}{B_{w,k}^T}$ as a limit as $h$ goes to 0 (which is $\int r_1 d\m_{w,T,k}$, see step 5 in the proof of theorem \ref{prop3.6}) and $\nr{B_{w,k}^T} = O(h^{-\frac 12})$, so: 
\begin{eqnarray} 
\label{eq3.83}
\lefteqn{\int (-H_{p}  + 2 \Im E_1 + 2V_2 )q \,  d\m_{w,T,k} }\\
\nonumber
&& = \lim _{h \to 0}   \innp{  \Opw(- H_p(q) + 2\Im E_1 q + 2 V_2 q) B_{w,k}^T}{B_{w,k}^T}\\
\nonumber
&& = \lim _{h \to 0}  \innp{ - \frac ih [\huh,\Opw(q)]   + 2\Im E_1 \Opw(q) + 2 V_2 \Opw(q) B_{w,k}^T}{B_{w,k}^T} \\
\nonumber
&& = -  \lim_{h\to 0} \frac i h \innp{((\hh-E_h)^* \Opw(q) - \Opw(q)(\hh-E_h))B_{w,k}^T} {B_{w,k}^T} \\
\nonumber
&& = - \lim_{h\to 0} \frac i h \left( \innp{\Opw(q) B_{w,k}^T} {(\hh-E_h)B_{w,k}^T}-\innp {(\hh-E_h)B_{w,k}^T}{\Opw(q) B_{w,k}^T} \right) \\
\nonumber
&& = 0
\end{eqnarray}
according to proposition \ref{prop-b2nul}.\\

\noindent {\bf 4.} 
Let $q \in C_0^\infty(\R^{2n})$ and $\e > 0$. There exists $T \geq 0$ such that:
\[
\int q \, d\m_T \geq \int q \,  d\m - \frac \e 2
\]
We can find a finite number of $w_i \in \R^{2n}$ such that $\supp q \subset  \cup \Vc_{w_i,T}$ and either $w_i \in \negg \cup \L$ or $\Vc_{w_i,T} \cap (\negg \cup \L) = \emptyset$. With a partition of unity, we can write $q = \sum q_i$ with ${\supp q_i \subset \Vc_{w_i,T}}$ and show the result for each $q_i$. So without loss of generality we can assume that $\supp q \subset \Vc_{w,T}$ for some $w \in \negg \cup \L$. According to \eqref{eq3.83} we have:
\[
\begin{aligned}
\int (-H_{p}  + 2 \Im E_1 + 2V_2 )q \, d\m_T
& = \sum_{j=0}^{K_w^T} \int (-H_{p}  + 2 \Im E_1 + 2V_2 )q \,  d\m_{w,T,k}\\
& = \int (-H_{p}  + 2 \Im E_1 + 2V_2 )q \,  d\m_{w,T,0}
\end{aligned}
\]

This is zero unless $w \in \negg$, which we now assume. Let $g \in C_0^\infty(\R)$ supported in $]-\infty,1]$ with $g = 1 $ near 0. For $m\in\N$ and $(x,\x) \in \tilde \G (\t_0) \times \R^n $ we set $g_m(x,\x) = g (m t_x)$. In particular the function $(1-g_m)q$ vanishes near $\negg$, so:
\begin{equation*}
\int (-H_{p} + 2 V_2 +  2 \Im E_1)(1- g_m)q \, d\m  = 0
\end{equation*}
Then since $g_m$ is supported in $\tilde \G(0,\t_0)$ for all $m\in\N$, we can use \eqref{chgt-var} to have:
\begin{eqnarray*}
\lefteqn{\int_{\R^{2n}}  (-H_{p}  + 2 \Im E_1 + 2V_2 )q \, d\m_{w,T,0}}\\
&& = \int_{\R^{2n}}  (-H_{p}  + 2 \Im E_1 + 2V_2 )q g_m \, d\m_{w,T,0}\\
&& = \int_{\tilde \G (0,\t_0)} (-H_p + 2 \Im E_1 + 2V_2) (q g_m) (x,\partial\p(x)) \abs{b_0(x)}^2 \,dx\\
&& = 2 ^{n-d} \int_{0}^{\t_0} \hspace{-0.2cm} \int_{\negg} t^{n-d-1} \abs \x \Big(1+\bigo t 0(t)\Big)  \abs{b_0(\bar x (t,z,\x))}^2\\
&& \hspace{3cm} \times (-H_p + 2 \Im E_1 + 2V_2) (q g_m) (\bar x (t,z,\x),\partial\p(\bar x (t,z,\x)))   \, d\snegg(z,\x)\,dt
\end{eqnarray*}
According to \eqref{phi-et-flot} we have $(x,\partial \p (x)) = \vf^{t_x}(z_x,\x_x)$. On the other hand, by \eqref{symb-b0} and \eqref{det-hess} we have:
\begin{equation} \label{def-c}
2 ^{n-d} t^{n-d-1} \abs \x \abs{b_0(\bar x (t,z,\x))}^2 \limt {t} 0 \pi (2\pi)^{d-n} A(z)^2 \abs \x \inv \hat S(\x)^2 =: c(z,\x)
\end{equation}
so:
\begin{eqnarray*}
\lefteqn{\int_{\R^{2n}}  (-H_{p}  + 2 \Im E_1 + 2V_2 )q \, d\m_{w,T,0}}\\
&& = -  \int _0^{\t_0 }  \hspace{-0.2cm}\int_{\negg} (\partial_t - 2 \Im E_1 - 2V_2)( q (\vf^t(z,\x)) g(mt)) c(z,\x)  \Big(1+\bigo t 0(t)\Big)  \, d\snegg(z,\x) \, dt \\
&& = -  \int _0^{\t_0 } \hspace{-0.2cm} \int_{ \negg} g(t m) (\partial _t  - 2 \Im E_1 - 2V_2) (q(\vf^t(z,\x))) c(z,\x)  \Big(1+\bigo t 0(t)\Big)  \, d\snegg(z,\x) \, dt\\
&& \quad - \int _0^{\t_0 }  \hspace{-0.2cm}\int_{ \negg} m g'(tm) q (\vf^t(z,\x)) c(z,\x)  \Big(1+\bigo t 0(t)\Big)  \, d\snegg(z,\x) \, dt 
\end{eqnarray*}
and hence:
\begin{eqnarray*}
\lefteqn{ \abs{\int  (-H_{p}  + 2 \Im E_1 + 2V_2 )q  \,  d\m_{w,T,0}- \int_{ \negg} q (z,\x)  c(z,\x)    \, d\snegg(z,\x) }}\\
&& \leq O \left( \frac 1m \right) + \abs {\int _0^{\t_0 } \hspace{-0.2cm} \int_{ \negg}  m g'(tm) \big( q(z,\x) - q (\vf^t(z,\x))\big) c(z,\x)  \, d\snegg(z,\x) \, dt} \\
&& \leq O \left( \frac 1m \right) + \int _0^{\t_0 } \hspace{-0.2cm} \int_{ \negg}  m \abs{ g'(tm) } \sup_{0\leq t \leq \frac 1m} \abs{q(z,\x) - q (\vf^t(z,\x)) \big)} c(z,\x)  \, d\snegg(z,\x) \, dt \\
&& = O\left( \frac 1m \right)  
\end{eqnarray*}
It only remains to choose $m$ so large that the rest is less than $\frac \e 2$. 
\end{proof}

As said in the introduction, $\m$ is actually characterized by the three properties of theorem \ref{th2.1} and is given by \eqref{expr-mu}:

\begin{proposition} \label{prop-caract}
Let $\nu$ be a Radon measure on $\R^{2n}$ which satisfies the three properties of theorem \ref{th2.1}. Then for all $q \in C_0^\infty(\R^{2n})$ we have:
\begin{equation} \label{expr-nu}
\int_{\R^{2n}} q \, d\nu = \int_0^{+\infty} \int_\negg c(z,\x)  q(\vf^t(z,\x)) e^{-2t\Im E_1 - 2\int_0^t V_2(\bar x (s,z,\x))\,ds} \,d\snegg(z,\x) \,dt
\end{equation}
where the function $c$ is defined in \eqref{def-c}.
\end{proposition}

\begin{proof}
Let $I_1$ be an open interval such that $\bar I \subset I_1 \subset \bar {I_1} \subset J$. Let $q \in C_0^\infty(\R^{2n})$. According to property (i), if $\supp q \subset p\inv(\R \setminus I)$ then $\int q \, d\nu = 0$ which is consistent with \eqref{expr-nu}, since both sides are zero. So we can assume that $\supp q \subset p\inv(I_1)$.

Using property (iii) we see that:
\begin{eqnarray*}
\lefteqn{\frac d {dt} \int_{\R^{2n}} (q \circ \vf^t) e^{-2t\Im E_1 - 2\int_0^t V_2 \circ \vf^{t-s} \,ds } \, d\nu}\\
&& = \int_{\R^{2n}} (H_p-2\Im E_1 - 2V_2) \left( (q \circ \vf^t) e^{-2t\Im E_1 - 2\int_0^t V_2\circ \vf^{t-s} \,ds} \right)  \,d\nu\\
&& = -\int_\negg c(z,\x)   \left( (q \circ \vf^t) e^{-2t\Im E_1 - 2\int_0^t V_2\circ \vf^{t-s} \,ds} \right)(z,\x) \,d\snegg(z,\x)
\end{eqnarray*}
and hence, for all $\t \geq 0$:
\[
\begin{aligned}
\int_{\R^{2n}} q \, d\m
& = \int_{\R^{2n}} (q \circ \vf^\t) e^{-2 \t \Im E_1 - 2\int_0^\t V_2\circ \vf^{\t-s} \,ds}\, d\nu\\
& \quad + \int_0^\t \int_\negg c(z,\x)   \left( (q \circ \vf^t) e^{-2 t\Im E_1 - 2\int_0^t V_2\circ \vf^{t-s} \,ds} \right)(z,\x) \,d\snegg(z,\x) \, dt
\end{aligned}
\]
So we only have to prove that:
\[
 \int_{\R^{2n}} (q \circ \vf^\t) e^{-2\t\Im E_1 - 2\int_0^\t V_2\circ \vf^{\t-s} \,ds}\, d\nu \limt \t {+\infty} 0
\]
For $R \geq 0$ we set: $K_R = p\inv(\bar I_1) \cap B_x(R)$. According to property (ii), we can find $R \geq 0$ such that $\nu$ vanishes on $\G_-(R,- \frac 1 2)$ and:
\[
\bigcup _{t\geq 0} \supp (q\circ \vf^t) \subset \G_-\left( R, -\frac 12 \right) \cup K_R
\]
Let $\h \in C_0^\infty(\R^{2n})$ supported in $p\inv(J)$ and equal to 1 on $K_R$. For $\t \geq 0$, since $\nu$ vanishes on $\G_- \left(R,-\frac 12 \right)$:
\[
\int_{\R^{2n}} (q \circ \vf^\t) e^{-2t\Im E_1 - 2\int_0^\t V_2\circ \vf^{\t-s} \,ds}\, d\nu =  \int_{\R^{2n}} \h (q \circ \vf^\t) e^{-2t\Im E_1 - 2\int_0^\t V_2\circ \vf^{\t-s} \,ds}\, d\nu
\]
As $\nu$ is a Radon measure, there is a constant $C\geq 0$ such that for all $\tilde q \in C_0^\infty(\R^{2n})$ with $\supp q \subset \supp \h$ we have:
\[
\abs{\int_{\R^{2n}} \tilde q \, d\nu} \leq C \nr {\tilde q}_{L^\infty(\R^{2n})}
\]
so we only need to prove that:
\[
\sup_{w \in \R^{2n}} \abs{\h(w) (q \circ \vf^\t)(w) e^{-2\t\Im E_1 - 2\int_0^\t (V_2\circ \vf^{\t-s})(w) \,ds}} \limt \t {+\infty} 0
\]
This is clear if $\Im E_1 > 0$. Otherwise, this can be done with lemma \ref{prop-amortis} as in the proof of proposition \ref{prop-super-egorov}.
\end{proof}

\section{Estimate of the outgoing solution in the incoming region}  \label{sec-incoming}

The theorem we want to prove in this section is the following:

\begin{theorem} \label{estimrobert}
Let $N \in \N$ and $E_h = E_0 + O(h)$ be an energy such that for all $h \in ]0,h_0]$, $\Im E_h > 0$ or $E_h$ satisfies \eqref{hyp2}. Let $d > 0$ and $\s \in ]0, 1[$. Then there exits $\n\in\N$ and $R >0$ such that if the symbols $\o_+,\o\in \symb_0$ have supports in $\G_+( R, d ,\s)$ (respectively outside $\G_+( R_1,  d_1, \s_1)$ with $ R_1 <  R$, $ d_1 < d$ and $ \s_1 < \s$) then for all $\a > \frac 12$ we have:
\begin{eqnarray}
\label{estimrob2} & \nr{\pppg x ^{-\a} \Op(\o) (\hh - (E_h+i0))\inv \Op(\o_+) \pppg x ^{-\n}} = \bigo h 0 (h^N)
\end{eqnarray}
Similarly, if $\supp \o_- \subset \G_-(R,d,-\s)$ and $\supp \o\cap \G_-(R_1,d_1,-\s_1)=\emptyset$ then:
\begin{equation} \label{estimrob4}  
\nr{\pppg x ^{-\a} \Op(\o) (\hh^* - (E_h-i0))\inv \Op(\o_-) \pppg x ^{-\n}} =\bigo h 0 (h^N)
\end{equation}
\end{theorem}

\begin{remark*}
This is the analog of lemma 2.3 in \cite{robertt89} in the dissipative case. Note that here $\nu$ is different from $\a$ and may be large.
\end{remark*}

\begin{remark*} 
Taking the adjoint in \eqref{estimrob4} gives:
\begin{equation*} 
\nr{ \pppg x ^{-\n}\Op(\o_-) (\hh - (E_h+i0))\inv \Op(\o) \pppg x ^{-\a}} = \bigo h 0 (h^N)
\end{equation*}
which proves proposition \ref{prop-incoming}. This theorem proves that the solution $u_h =(\hh - (E+i0))\inv S_h$ is microlocally zero in the incoming region.
\end{remark*}

To prove this theorem we follow \cite{wang88}. In particular we use the following result taken from \cite{isozakik85}:

\begin{proposition}  \label{isozakik}
Let $d_0 \in ]0,d_1[$ and $\s_0 \in ]0,\s_1[$. There exists $R_0 > 0$ and ${\vf_\pm \in C^\infty (\R ^{2n})}$ satisfying:
\begin{equation} \label{prophi1}
\forall (x,\x) \in \G_\pm(R_0,d_0,\pm\s_0),\quad \abs{\nabla_x \vf_\pm(x,\x)}^2 + V_1(x) = \abs \x ^2
\end{equation}
and:
\begin{eqnarray}
\label{prophi2}
\forall (x,\x) \in \R^{2n},\forall \a,\b \in \N^n, & & \abs{\partial_x^\a \partial_\x^\b ( \vf_\pm(x,\x) - \innp x \x)} \leq C_{\a,\b}\pppg x ^{1 - \rho - \abs \a}
\end{eqnarray}
for some $\rho >0$.
\end{proposition}

Without loss of generality we may assume that this is the same constant $\rho$ as in \eqref{h1}.

\begin{remark*} 
As mentioned in \cite{wang88} (see (2.4)), we can assume that the constants $C_{\a,\b}$ in \eqref{prophi2} are as small as we wish if we take $R$ large enough. Indeed, if we take a function $\h \in C^\infty(\R^n)$ such that $\h(x) =0$ if $\abs x \leq \frac 12$ and $\h(x) = 1$ if $\abs x \geq 1$, and, for $R > R_0$:
\begin{equation} \label{phir}
\vf_{R,\pm} : (x,\x) \mapsto (\vf_\pm(x,\x) - \innp x \x) \h \left( \frac x R \right ) + \innp x \x
\end{equation}
Then:
\begin{equation} \label{prophi1r}
\forall (x,\x) \in \G_\pm(R,d_0,\s_0),\quad \abs{\nabla_x \vf_{R,\pm}(x,\x)}^2 + V_1(x) = \abs \x ^2
\end{equation}
and for any $\rho_1,\rho_2 >0$ such that $\rho = \rho_1+\rho_2$:
\begin{equation} \label{phi-innp}
\forall ( x,\x) \in \R^{2n}, \quad \abs{ \partial_x^\a \partial_\x^\b ( \vf_{R,\pm}(x,\x) - \innp x \x)} \leq C_{\a, \b}  R^{-\rho_1} \pppg x ^{1-\rho_2 - \abs \a}
\end{equation}
where $C_{\a,\b}$ does not depend on $R$.
\end{remark*}

We are going to use the Fourier integral operators $I_h(a,\vf)$ defined as follows:
\[
I_h(a,\vf) u (x) = \frac 1 {(2\pi h)^n} \int_{\R^n}\int_{\R^n} e^{\frac ih (\vf(x,\x)- \innp y \x)} a(x,\x) u(y) \, dy\,d\x
\]

As in \cite{wang88}, the idea of the proof is to find two symbols $a$ and $e$ such that:
\[
U_h(t)I_h(a,\vf) \approx I_h(a,\vf) U_0^h(t) 
\quad \text{and} \quad
\Op(\o_+) \approx I_h(a,\vf)I_h(e,\vf)^*
\]
when $h$ goes to 0. For a short range absorption coefficient $V_2$, we can actually do as in \cite{wang88}, but in the long range case, we have to consider a time dependant symbol $a(t,h)$. In this situation we have:
\begin{eqnarray}
\label{diff-int}
\lefteqn{U_h(t)I_h(a(t,h),\vf_\pm) - I_h(a(t,h),\vf_\pm) U_0^h(t)}\\
\nonumber
&& = \int_0^t U_h(t) \left( -\frac ih \hh I_h(a(s,h),\vf_\pm) + I_h(\partial_t a(s,h),\vf_\pm) + \frac ih I_h(a(s,h),\vf_\pm) \hoh\right) U_0^h(t-s)\, ds
\end{eqnarray}

\begin{proposition} \label{propap}
Let $a(t,h) \in \symbor$ be a time-dependant symbol, $\vf=\vf_+$ or $\vf_-$ given by proposition \ref{isozakik} and $h\in]0,1]$. Then we have:
\begin{equation*}
-\frac ih \hh I_h(a(t,h),\vf) + I_h(\partial_t a(t,h),\vf) +\frac ih I_h(a(t,h),\vf) \hoh = I_h(p(t,h),\vf)
\end{equation*}
where:
\begin{eqnarray} \label{defp}
\lefteqn{p(t,h)}\\
\nonumber && = -\frac ih (\abs{\partial_x \vf}^2 + V_1 - \x^2) a(t,h) + \Big(\partial_t a(t,h) - 2\partial_x a(t,h) . \partial_x \vf - a(t,h) \D_x \vf - a(t,h) V_2\Big)\\
\nonumber && \quad  +  ih  \D_x a(t,h)
\end{eqnarray}
\end{proposition}

\begin{remark*}
If moreover $a(t,h)$ is of the form:
\begin{equation*} 
a(t,h)= \sum_{j=0}^N h^j a_j(t) 
\end{equation*}
with $a_j\in\symbor$ for all $j\in \Ii 0 N$, then $p(t,h)$ takes the form:
\begin{equation*} 
\begin{aligned}
p(t,h) =
& -\frac ih (\abs{\partial_x \vf}^2 + V_1 - \x^2) a(t,h) 
 + \Big(\partial_t a_0(t,h) - 2 \partial_x a_0(t) . \partial_x \vf - a_0(t) \D_x \vf - a_0(t) V_2 \Big)\\
& + \sum_{j=1}^N h^{j} \Big(\partial_t a_j(t,h) - 2\partial_x a_j(t) . \partial_x \vf - a_j(t) \D_x \vf - a_j(t) V_2 + i \D_x a_{j-1}(t)\Big) \\
& + i h^{N+1} \D_x a_N(t)
\end{aligned}
\end{equation*}
This gives the transport equations the symbols $a_j$ have to satisfy if we want $I_h(p(t,h),\vf) = \bigo h 0 (h^{N+1})$.
\end{remark*}

\begin{remark*}
Similarly we have:
\begin{equation*}
-\frac ih \hh^* I_h(a(t,h),\vf) + I_h(\partial_t a(t,h),\vf) + \frac ih I_h(a(t,h),\vf) \hoh = I_h(p_*(t,h),\vf)
\end{equation*}
where:
\begin{eqnarray*}
\lefteqn{p_*(t,h)}\\
&& = -\frac ih (\abs{\partial_x \vf}^2 + V_1 - \x^2) a(t,h) + \Big(\partial_t a(t,h) - 2\partial_x a(t,h) . \partial_x \vf - a(t,h) \D_x \vf + a(t,h) V_2\Big)\\
&& \quad  +  ih  \D_x a(t,h)
\end{eqnarray*}
\end{remark*}

\begin{lemma} \label{flotphi} Let $\vf$ be a function which satisfies (\ref{prophi2}). Then for all $(x,\x) \in \R^{2n}$, the Cauchy problem:
\begin{equation*}
\left\{ \begin{array}l
\frac {\partial r}{\partial t}(t,x,\x) = \partial_x \vf (r(t,x,\x),\x)\\ r(0,x,\x) = x
\end{array} \right.
\end{equation*}
has a unique solution defined on $\R$. Furthermore, for $\gamma \in]0,\s_1[$, if $R$ is large enough, we have the following properties:
\begin{enumerate}[(i)]
\item For $(x,\x) \in \G_\pm (d_1,\pm\s_1)$ and $\pm t\geq 0$ we have:
\begin{equation} \label{minflot2}
\abs{r(t,x,\x)} \geq \abs x + (\s_1-\g) d_1  \abs t 
\end{equation} 
\item  For $(x,\x) \in \G_\pm (d_1,\pm\s_1)$, $\pm t\geq 0$ and $\abs \a + \abs \b \geq 1$, there is a constant $c_{\a,\b}$ such that:
\begin{equation} \label{estimderflot}
\abs{ \partial_x^\a \partial_\x ^\b r(t,x,\x)} \leq c_ {\a,\b} \max(\abs t , \pppg x ) \pppg x ^{ - \abs\a}
\end{equation}
\end{enumerate}
\end{lemma}

\begin{proof}
Let $(x,\x) \in \R^{2n}$. We have:
\begin{equation} \label{exprflot}
r(t,x,\x) = x + t\x + \int _0 ^t (\partial_x \vf(r(s,x,\x),\x)-  \x) \, ds
\end{equation}
where $r(\cdot,x,\x)$ is defined, that is everywhere since $(\partial_x \vf (r(t,x,\x),\x) - \x)$ is bounded according to \eqref{prophi2}.\\

\emph{(i)} By \eqref{phi-innp}, if $R$ is large enough we can assume that:
\begin{equation*}
\forall (x,\x) \in \R^{2n}, \quad \abs{\partial_x \vf (x,\x) - \x } \leq \gamma d_1
\end{equation*}
and hence:
\begin{equation*} 
\abs{r(t,x,\x) - x - t\x} \leq  \abs t \gamma d_1
\end{equation*}
If $(x,\x) \in \G_\pm(d_1,\pm\s_1)$ and $\pm t\geq 0$, then:
\begin{equation*}
\abs{x +t \x} \geq \frac 1 {\abs x}\innp {x} {x+t\x} \geq \abs x  + \s_1 \abs t   \abs \x \geq \abs x + \abs t \s_1 d_1
\end{equation*}
so:
\begin{equation*} 
\abs{r(t,x,\x)}  \geq \abs{x+t\x} - \gamma \abs t d_1 \geq \abs x + (\s_1 - \g) d_1 \abs t
\end{equation*}
which proves \eqref{minflot2}.\\

\emph{(ii)} We prove \eqref{estimderflot} by induction on $\abs\a + \abs \b$, beginning by the case $\abs \a = 1$, $\b= 0$. Let $\pm t \geq 0$ and $(x,\x)\in \G_+(d_1,\s_1)$. We have:
\begin{equation*}
\begin{aligned}
\partial_t\partial_x r(t,x,\x) = \partial^2_x \vf(r(t,x,\x),\x).\partial_x r(t,x,\x)
\end{aligned}
\end{equation*}
According to Gronwall lemma, \eqref{prophi2} and \eqref{minflot2}, we obtain the estimate:
\[
\begin{aligned}
\nr{\partial_x r(t,x,\x) }
& \leq \exp\left( \int_0^t \nr{\partial_x^2 \vf (r(s,x,\x),\x)} \, ds \right) \leq \exp\left( \int_0^t c \pppg{r(s,x,\x)}^{-1-\rho} \, ds \right)\\
& \leq \exp\left( \int_0^t c \pppg s ^{-1-\rho} \, ds \right) \leq c \leq c \max (\abs t,\pppg x) \pppg x \inv
\end{aligned}
\]

Similarly, if $\a=0$ and $\abs \b = 1$ we have:
\begin{equation*}
\begin{aligned}
\partial_t  \partial_\x r(t,x,\x) 
 = \partial_x^2 \vf(r(t,x,\x),\x).\partial_\x r(t,x,\x) +  \partial_{x} \partial_\x \vf(r(t,x,\x),\x) 
\end{aligned}
\end{equation*}
and then:
\begin{equation*}
\begin{aligned}
\nr{\partial_t  \partial_\x r(t,x,\x) }  \leq  \abs{\int_{s = 0}^t \nr{\partial_{x} \partial_\x \vf(r(s,x,\x),\x) } \exp \left( \int_{\t=s}^t \nr{ \partial^2_x \vf(r(\t,x,\x),\x)}\, d\t\right)\, ds}  \leq c  \abs t 
\end{aligned}
\end{equation*}

We now assume that we have proved \eqref{estimderflot} for $1 \leq \abs \a + \abs \b \leq k \in\N^*$ and we consider $\a$ and $\b$ such that $\abs \a + \abs \b = k+ 1$. For $j\in\Ii 1 n$ we have:
\begin{equation*}
\begin{aligned}
\partial_t \partial_x^\a \partial_\x^\b r_j(t,x,\x)  
& = \partial_x^\a \partial_\x^\b (\partial_{x_j} \vf(r(t,x,\x),\x)) \\
& = \sum_{l=1}^n \partial_{x_l,x_j}^2 \vf(r(t,x,\x),\x) \, \partial_x^\a \partial_\x^\b r_l(t,x,\x) + B_j(t,x,\x)
\end{aligned}
\end{equation*}
where $B_j$ is a sum of terms of the form:
\begin{equation*} \label{terme-bj}
(\partial_x^{\g}\partial_\x^\d  \partial_{x_j}\vf)(r(t,x,\x),\x) \prod_{s=1}^{\abs \g} (\partial^{\a_s}_x \partial_\x^{\b_s} r_{l_s})(t,x,\x)
\end{equation*}
with $\abs \gamma + \abs \d \geq 2$ and for all $s$ : $l_s  \in \Ii 1 n$, $\abs{\a_s} + \abs {\b_s} \leq k$, $\sum \a_s=\a$ and $\d + \sum \b_s= \b$. Then $B_j$ is smaller than:
\[
\pppg {r(t,x,\x)}^{-\abs \g- \rho} \prod_{s= 1} ^{\abs \g} \max ( \abs t, \pppg x) \pppg x ^{-\abs {\a_s}} \leq c \pppg x ^{-\a} 
\]
and finally \eqref{estimderflot} holds since:
\begin{equation*}
\nr{\partial_t \partial_x^\a \partial_\x^\b r(t,x,\x)} \leq  \abs{\int_{s = 0}^t \nr{B(t,x,\x)} \exp \left( \int_{\t=s}^t \nr{ \partial^2_x \vf(r(\t,x,\x),\x)}\, d\t\right)\, ds}\leq c \abs t \pppg x ^{-\a}
\end{equation*}
\end{proof}

Let $r_\pm$ be the functions defined in this proposition for $\vf = \vf_\pm$ and:
\begin{equation*} 
F_\pm(t,x,\x) =  \D_x \vf_\pm(r_\pm(t,x,\x),\x) \pm V_2 (r_\pm(t,x,\x))
\end{equation*}
In particular we have:
\begin{equation*}
F_\pm(0,x,\x) =  \D_x \vf_\pm (x,\x) \pm V_2(x)
\quad \text{and} \quad 
F_\pm(t,r_\pm(s,x,\x),\x) = F_\pm(t+s,x,\x)
\end{equation*}

\begin{proposition} \label{prop-eqtrans}
The functions $a_{j,\pm}(t,h), j\in \N$ defined by:
\begin{equation*} 
a_{0,\pm}(t,x,\x) = \exp\left({- \int _{s=0}^t \left(F_\pm(2s,x,\x)\right)}\, ds \right)
\end{equation*}
and for $j \geq 1$:
\begin{equation*}  
a_{j,\pm}(t,x,\x) = i \int_{\t = 0 }^t \D_x a_{j-1,\pm}(\t,r_\pm(2\t,x,\x),\x )a_0(\t,x,\x) \, d\t
\end{equation*}
are solutions of the transport equations:
\begin{equation}  \label{transpa0}
\partial_t a_{0,\pm}(t,h) - 2\partial_x a_{0,\pm}(t) . \partial_x \vf_\pm - a_{0,\pm}(t) \D_x \vf_\pm   \mp a_{0,\pm}(t) V_2 = 0
\end{equation}
and for $j\geq 1$:
\begin{equation} \label{transpaj}
\partial_t a_{j,\pm}(t,h) - 2 \partial_x a_{j,\pm}(t) . \partial_x \vf_\pm - a_{j,\pm}(t) \D_x \vf_\pm  \mp a_{j,\pm}(t) V_2 + i \D_x a_{j-1}(t) = 0
\end{equation}
and satisfy estimates:
\begin{equation} \label{estima0}
\text{for } \pm t \geq 0,  (x,\x) \in \G_\pm(d_1, \pm \s_1),\quad \abs{\partial_x^\a \partial_\x^\b a_{j,\pm}(t,x,\x)} \leq c_{\a,\b} \abs{t} ^{j+(\abs \a + \abs \b)(1-\rho)} \pppg x ^{- \abs \a}
\end{equation}
\end{proposition}

\begin{proof}
We prove \eqref{estima0}. For $\a,\b \in \N^n$, the derivative $\partial_x^\a \partial_\x^\b a_{0,\pm}(t,x,\x,h)$ is a sum of terms of the form:
\begin{equation*}
\prod_{k=1}^J \partial_x^{\m_k} \partial_\x^{\n_k} \left( \int_0^t F_\pm(2s,x,\x) \, ds \right) a_{0,\pm}(t,x,\x)
\end{equation*}
with $\sum \m_k = \a$, $\sum \n_k = \b$ and for all $k \in \Ii 1 J$: $\abs {\m_k} + \abs {\n_k} \geq 1$ (and in particular ${J\leq \abs \a + \abs \b}$).
We first remark that according to \eqref{prophi2} and \eqref{minflot2} together with non\-ne\-ga\-ti\-veness of $V_2$ the symbol $a_0$ is bounded uniformly in $\pm t \geq 0$. Hence we have to prove:
\begin{equation*}
\abs{\int_0^t \partial_x^{\m_k} \partial_\x^{\n_k} F_\pm(2s,x,\x) \, ds }  \leq c_{\a,\b} \abs t^{(\abs{\m_k}+\abs {\n_k})(1-\rho)}  \pppg x ^{-\abs{\m_k}}
\end{equation*}

Let $\pm t\geq 0$, $(x,\x) \in \G_\pm(d_1,\pm\s_1)$ and $\m,\n \in \N^n$. Then:
\begin{equation*}
\partial_x^\m \partial_\x^\n \left(\int_0 ^t  F_\pm(2s,x,\x) \, ds \right)
\end{equation*}
is a sum of terms of the form:
\begin{equation}  \label{estima0terme}
\int_0^t \partial_x^\d \partial_\x^\l(\D_x\vf_\pm + V_2) (r_\pm(2s,x,\x),\x) \prod_{k=1}^{\abs{\d}} \partial_x^{\m_k} \partial_\x^{\n_k} r_\pm (2s,x,\x)\, ds 
\end{equation}
with $\sum_{j=1}^{\abs \d} \m_k = \m$, $\sum_{j=1}^{\abs{\d}}  \n_k + \l = \n$ and for all $k \in \Ii 1{\abs \d}$: $\abs {\m_k} + \abs{\n_k} \geq 1$. 
By \eqref{h1}, \eqref{prophi2} and \eqref{estimderflot} we have:
\[
\begin{aligned}
\abs{\partial_x^\m \partial_\x^\n \left(\int_0 ^t  F_\pm(s,x,\x) \, ds \right)}
& \leq c \abs t^{1-\rho}  \pppg x ^{-\abs {\m}}
\end{aligned}
\]
this proves \eqref{estima0} for $j=0$. We now prove the general case by induction. For $\a,\b \in \N^n$ the derivative $\partial_x^\a \partial_\x^\b a_{j+1,\pm}(t,x,\x)$ is a sum of terms of the form:
\[
i \int _{\t = 0}^t  \partial_x^\m \partial_\x^\n (\D_x a_{j,\pm}(t,r_\pm(2(\t-t),x,\x),\x)) \times \partial_x^{\a-\m} \partial_\x^{\b - \n} a_{0,\pm}(\t,x,\x)  \,d\t
\]
We already know that for $\t \in [0,t]$:
\[
\abs{\partial_x^{\a-\m} \partial_\x^{\b - \n}a_{0,\pm}(\t,x,\x)} \leq c \abs t^{(1-\rho) (\abs{\a - \m} + \abs{\b - \n})} \pppg x ^{-\abs{\a - \m}}
\]
So it remains to show:
\begin{equation*} 
 \abs{ \partial_x^\m \partial_\x^\n (\D_x a_{j,\pm}(\t,r_\pm(2\t,x,\x),\x))}  \leq c \, \abs t^{j+(1-\rho)(\abs \m + \abs \n)}\pppg x ^{-\abs \m}
\end{equation*}

But $\partial_x^\m \partial_\x^\n (\D_x a_{j,\pm}(\t,r_\pm(2\t,x,\x),\x))$ is a sum of terms of the form:
\[
(\partial_x^\d \partial_\x^\l \D_x a_{j,\pm}) (t,r_\pm(2\t,x,\x),\x) \prod_{k= 1}^{\abs{\d}} (\partial_x^{\m_k} \partial_\x^{\n_k} r_\pm)(2\t,x,\x)
\]
with $\m = \sum_{k=1}^{\abs \d} \m_k$ and $\n = \l + \sum_{k=1}^{\abs \d} \n_k$, and:
\begin{eqnarray*}
\lefteqn{ \abs{(\partial_x^\d \partial_\x^\l \D_x a_{j,\pm}) (\t,r_\pm(2\t,x,\x),\x) \prod_{j= 1}^{\abs{\d}} (\partial_x^{\m_j} \partial_\x^{\n_j} r_\pm)(2\t,x,\x)}}\\
&& \leq c  \abs \t^{j+ (1-\rho)(\abs \d + \abs \l + 2)} \pppg {r_\pm(2\t,x,\x)} ^{-\abs \d-2}  \max(\abs{2\t},\pppg x) ^{\d}  \pppg x ^{- \sum_{j=1}^{\abs\d} \m_j} \\
&& \leq c \, \abs t ^{j + (1-\rho)(\abs \d + \abs \l)} \pppg x ^{-\abs \m}
\end{eqnarray*}
which concludes the proof after integration over $\t\in[0,t]$.
\end{proof}

\begin{remark*} 
This is for this part of the proof that we need a time-dependant symbol. Indeed, following exactly the proof of \cite{wang88} would have led to consider:
\[
a_0(x,\x) = \exp\left(\int_0^\infty F(t,x,\x)\, dt\right)
\]
which may have no sense for a long range imaginary part of the potential $V_2$. For a short range potential we do not have such a problem and the sign of $V_2$ we have used here does not matter.
\end{remark*}

Let $\s_2$ and $\s_3$ such that $\s_1 < \s_2 < \s_3<\s$, $R_2$ and $R_3$ such that $R_1 < R_2 <R_3<R$ and $d_2,d_3$ such that $d_1 < d_2 < d_3 < d$. We consider functions $\rho_1 \in C^\infty (\R)$ such that $\rho_1(s) = 0$ if $s\leq \s_2$ and 1 if $s \geq \s_3$, $\rho_2 \in C^\infty (\R)$ such that $\rho_2(s) = 0$ and $s \leq d_2$ and 1 if $s \geq d_3$ and $\rho_3 \in C^\infty (\R)$ such that $\rho_3(s) = 0$ if $s\leq R_2$ and $\rho_3(s) =1$ if $s \geq R_3$. Then we set:

\begin{equation*} 
b_{\pm}(t,x,\x,h) = \p_{\pm}(x,\x) \sum_{j=0}^N h^j a_{j,\pm}(t,x,\x) 
\quad \text{where: }
\p_\pm(x,\x) =  \rho_1 \left( \frac {\pm\innp x \x}{\abs x \abs \x} \right) \rho_ 2 (\abs \x) \rho_3(\abs x )
\end{equation*}
We also set:
\begin{equation*} 
\begin{aligned}
p_{\pm}(t,h)
& = \frac ih (\abs{\partial_x \vf_\pm}^2 + V_1 - \x^2) b_{\pm}(t,h)  \\
& \quad + (\partial_t b_{\pm}(t,h) + 2\partial_x b_{\pm}(t,h) . \partial_x \vf _\pm + b_{\pm}(t,h) \D_x \vf_\pm \pm b_{\pm}(t,h) V_2) \\
& \quad -  ih^{N+1}  \D_x b_{\pm}(t,h)
\end{aligned}
\end{equation*}
as given by proposition \ref{propap}.

\begin{proposition} \label{estimbp}
The symbols $b_\pm$ and $p_\pm$ satisfy:
\begin{enumerate} [(i)]
\item $\supp b_\pm \subset \G_\pm(R_2,d_2,\pm\s_2)$ and for $\pm t \geq 0$, $(x,\x) \in \G_\pm(R_2,d_2,\pm\s_2)$ and $\a,\b \in \N^n$ we have:
\begin{equation} \label{estimb}
\abs{\partial_x^\a \partial_\x^\b b(t,x,\x,h)} \leq c_{\a,\b} \abs t ^{N+(\abs \a + \abs \b)(1-\rho)} \pppg x ^{-\abs \a}
\end{equation}
\item $\supp p_\pm \subset \G_\pm(R_2,d_2,\pm\s_2)$ and for $\pm t \geq 0$, $(x,\x) \in \G_\pm(R_2,d_2,\pm\s_2)$ and $\a,\b \in \N^n$ we have:
\begin{equation} \label{estimp}
\abs{\partial_x^\a \partial_\x^\b p_\pm(t,x,\x,h)} \leq c_{\a,\b} \abs t ^{N+(2+\abs \a + \abs \b)(1-\rho)} \pppg x ^{- \abs \a}
\end{equation}
If furthermore $(x,\x)\in \G_\pm(R_3,d_3,\pm\s_3)$ then we have:
\begin{equation} \label{estimph}
\abs{\partial_x^\a \partial_\x^\b p_\pm(t,x,\x,h)} \leq c_{\a,\b} h^{N+1} \abs t ^{N+(2+\abs \a + \abs \b)(1-\rho)} \pppg x ^{-2 - \abs \a}
\end{equation}
\end{enumerate}
\end{proposition}

\begin{proof}
\eqref{estimb} comes from \eqref{estima0}. According to \eqref{transpa0} and \eqref{transpaj} we have:
\[
p_\pm(t,x,\x,h) = 2\partial_x \p_\pm(x,\x).\partial_x \vf_\pm(x,\x) \sum_{j=0}^N a_{j,\pm}(t,x,\x)  -ih^{N+1} \D_x b_\pm(t,x,\x,h)
\]
so \eqref{estimp} is a consequence of \eqref{estima0} and \eqref{estimb}. Finally, it remains to remark that for ${\pm t \geq 0}$ and $(x,\x) \in \G_\pm(R_3,d_3,\pm \s_3)$ we have $p_\pm(t,h) = -ih^{N+1} \D_x b_{\pm}(t,h)$ to get \eqref{estimph} from \eqref{estimb}.
\end{proof}

\begin{proposition} \label{propbe}
Let $R_5 \in ]R_3,R[$, $d_5 \in ]d_3,d[$ and $\s_5 \in ]\s_3,\s[$. There exists a symbol $e_{\pm}(h)$ of the form $e_{\pm}(h) = \sum_{j=0}^N h^j f_{j,\pm}$ with $f_{j,\pm} \in \symb_{-j}$ and $\supp f_{j,\pm} \subset \G_\pm(R_5,d_5,\pm \s_5)$ such that:
\begin{equation*}
I_h(b_\pm(0,h),\vf) I_h(e_{\n,\pm}(h),\vf)^* = \o_\pm(x,hD) + h^{N+1} \Op(r_{\pm}(h))
\end{equation*}
where $r_{\pm} \in \symb_{-N}$ uniformly in $h$.
\end{proposition}

\begin{proof}
This is lemma 4.5 in \cite{wang88}. Note that $b_\pm(0,h)$ is just $\p_\pm$.
\end{proof}

\begin{proposition} \label{propdect}
For all $\d \in \R$, there is $\n \in \N$ such that for all $l\in\R$ and $\pm t \geq 0$ we have:
\begin{equation} \label{eqdect}
\nr{ \pppg x ^l I_h(b_\pm(t,h),\vf) U_0^h(t) I_h(e_\pm,\vf)^* \pppg x ^{-1-\n-l}} \leq c \pppg t ^{-\d}
\end{equation}
and:
\begin{equation} \label{eqdect2}
\nr{ \pppg x ^l I_h(p_\pm(t,h),\vf) U_0^h(t) I_h(e_\pm,\vf)^* \pppg x ^{-1- \n -l}} \leq c h^{N+1}  \pppg t ^{-\d}
\end{equation}
\end{proposition}

\begin{proof}
For $u \in \Sc(\R^n)$ we have:
\begin{equation*}
\begin{aligned}
I_h(b_\pm(t,h),\vf_\pm) & U_0^h(t)  I_h(e_\pm (h),\vf_\pm)^*  u (x)\\
& = \frac 1 {(2\pi h)^n}\int_y\int_\x e^{\frac ih \z_\pm(t,x,y,\x)} b_\pm(t,x,\x,h) \bar {e_\pm(y,\x,h)} u(y)\, d\x\, dy
\end{aligned}
\end{equation*}
with $\z_\pm(t,x,y,\x) = \vf_\pm(x,\x) - \vf_\pm(y,\x) -t \x^2$.
If $R$ is large enough then for $(y,\x) \in \supp e_\pm$ we have:
\begin{equation} \label{majdect}
\begin{aligned}
\abs {\partial_\x \vf_\pm(y,\x) + 2 t \x}
 \geq \innp {\partial_\x \vf_\pm (y,\x) + 2 t \x} {\hat y}
\geq \abs y - c \abs y ^{1-\rho} + 2 \s_5 \abs t \abs \x
 \geq c_0 ( \abs y + \abs t)
\end{aligned}
\end{equation}
for some $c_0 > 0$.\\

We consider the operator $L$ such that for $u\in\Sc(\R^{2n})$:
\begin{equation*}
Lu = ih \frac {(\partial_\x \vf_\pm(y,\x) + 2 t \x).\partial_\x u}{\abs{\partial_\x \vf_\pm(y,\x) + 2 t \x}^2}
\end{equation*}
Then we have:
\begin{equation*}
L^*v = ih \divg_\x . \left(\frac {\partial_\x \vf_\pm(y,\x) + 2 t \x}{\abs{\partial_\x \vf_\pm(y,\x) + 2 t \x}^2}\, v \right)
\end{equation*}
In particular $L\left( e^{-\frac ih (\vf_\pm(y,\x) + t\x^2}\right) =    e^{-\frac ih (\vf_\pm(y,\x) + t\x^2)}$ so for $\n\in\N$:
\begin{eqnarray*}
\lefteqn{I_h(b_\pm(t,h),\vf_\pm) U_0^h(t)I_h(e_\pm(h),\vf_\pm)^*  u (x)}\\
&& = \frac 1 {(2\pi h)^n}\int_y \int_\x e^{-\frac ih (\vf_\pm(y,\x) + t\x^2)} (L^*)^\n \left(e^{\frac ih \vf_\pm(x,\x)} b_\pm(t,x,\x,h) \bar {e_\pm(y,\x,h )}\right) u(y) \, d\x\, dy
\end{eqnarray*}

We can check by induction on $\n \in \N$ that:
\[
(L^*)^\n \left(e^{\frac ih \vf_\pm(x,\x)} b_\pm(t,x,\x,h) \bar {e_\pm(y,\x,h)}\right) = \sum_{j=1}^{J_\n} e^{\frac ih \vf_\pm(x,\x)} b_{\n,\pm}^j(t,x,\x,h) \bar {e_{\n,\pm}^j(y,\x,h)}
\]
for some $J_\n \in\N$ and for all $j\in\Ii 1 {J_\n}$ we have:
\[
\abs{\partial_x^\a \partial _\x^\b b_{\n,\pm}^j (t,x,\x,h)} \leq c_{\a,\b} \abs t ^{N-(\abs \a + \abs \b)(1-\rho) - \rho \n} \pppg x ^{\n-\abs \a}
\]
and $e_0 \in \Sc_0$:
Indeed, this is true for $\n= 0$ by \eqref{estimb} and if this is true for some $\n\in\N$ then for $j \in \Ii 1 {J_\n}$ we have to compute:
\begin{eqnarray*}
\lefteqn{ih \divg_\x \left( \frac {\partial_\x \vf_\pm(y,\x) + 2t\x}{\abs {\partial_\x \vf_\pm(y,\x) + 2t\x}^2}e^{\frac ih \vf_\pm(x,\x)} b_{\n,\pm}^j(t,x,\x,h) \bar{e_{\n,\pm}^j(y,\x,h)} \right)}\\
&& = ih \abs { \partial_\x \vf_\pm(y,\x) + 2t\x } ^{-2} \times e^{\frac ih \vf_\pm(x,\x)} \\
&& \quad  \times \bigg[
(\D_\x\vf_\pm (y,\x) + 2tn) b_{\n,\pm}^j(t,x,\x,h) \bar {e_{\n,\pm}^j(y,\x)} \\
&& \quad \quad \quad + 2 \frac{( \Hess_\x \vf_\pm(y,\x) + 2t I_n ).{(\partial_\x \vf_\pm (y,\x) + 2t\x)}^2}{\abs{\partial_\x \vf_\pm (y,\x) + 2t\x}^2} b_{\n,\pm}^j(t,x,\x,h) \bar {e_{\n,\pm}^j(y,\x)}  \\ 
&& \quad \quad \quad + \, \frac ih ( \partial_\x \vf_\pm(y,\x) + 2t\x) \partial_\x\vf_\pm(x,\x) b_{\n,\pm}^j(t,x,\x,h) \bar {e_{\n,\pm}^j(y,\x)}\\
&& \quad \quad \quad + \, ( \partial_\x \vf_\pm(y,\x) + 2t\x). \partial_\x b_{\n,\pm}^j(t,x,\x,h) \, \bar {e_{\n,\pm}^j(y,\x)}\\
&& \quad \quad \quad + \, b_{\n,\pm}^j(t,x,\x,h) \, ( \partial_\x \vf_\pm(y,\x) + 2t\x) . \partial_\x  \bar {e_{\n,\pm}^j(y,\x)}\hspace{3cm} \bigg]
%
%
%
\end{eqnarray*}
and check each term using \eqref{majdect}. Note that the factor $\pppg x ^\n$ in the estimate is due to the third term. We only gain a power $t ^{-\rho \n}$ at each iteration because of the fourth term and the fact that we have a bad estimate in $t$ for the derivatives of $b_{\n,\pm}$. Nonetheless, for all $\n \in \N$ we get:
\begin{equation} \label{decomp-Ib}
I_h(b_\pm(t,h),\vf_\pm) U_0^h(t)I_h(e_\pm(h),\vf_\pm)^* =\sum_{j=1}^{J_\n} I_h(b_{\n,\pm}^j(t,h),\vf_\pm) U_0^h(t)I_h(e_{\n,\pm}^j(h),\vf_\pm)^*
\end{equation}

For any $\n \in \N$, the two operators $U_0^h(t)$ and $I_h(e_{\n,\pm}(h),\vf_\pm)^*$ are uniformly bounded in $t$ and $h$ from $L^{2,1+l+\n}$ into itself.
The norm of $I_h(b_{\n,\pm}(t,h),\vf_\pm)$ from $L^{2,1+l+\n}$ to $L^{2,l}$ is estimated by a finite number of derivatives of $b_{\n,\pm}^j$, say $M$ (see \cite{wang88}). Then we have to choose $\n$ such that $N + M(1-\rho) -\n \rho \leq - \d$ to obtain \eqref{eqdect}.\\

To prove \eqref{eqdect2} we introduce a function $\h \in C^\infty(\R)$ such that $\h(s) = 0$ if $s \leq \s_3$ and $\h(s) = 1$ if $s \geq \s_{4} \in ]\s_3,\s_5[$. Then we write $p_{2,\pm}(t,x,\x,h) = p_\pm(t,x,\x,h) \h\left( \pm \frac {\innp x \x}{\abs x \abs \x} \right)$ and $p_{1,\pm}(t,x,\x,h) = p_\pm(t,x,\x,h) -p_{2,\pm}(t,x,\x,h)$. We have:
\begin{equation*}
\abs{\partial_x^\a \partial_\x^\b p_{2,\pm}(t,x,\x,h)} \leq c_{\a,\b} h^{N+1} \abs t ^{N + (2+\abs \a +\abs \b)(1-\rho)} \pppg x ^{-2-\abs \a}
\end{equation*}
The same argument as above proves \eqref{eqdect2} with $p_\pm$ replaced by $p_{2,\pm}$.

For $p_{1,\pm}$, we remark that for $(x,\x) \in \supp p_{1,\pm} \subset \R^{2n} \setminus \G_\pm(R_4,d_4,\pm\s_{4})$ and $(y,\x)\in\supp e_\pm \subset \G_\pm(R_5,d_5,\pm\s_5)$ we have:
\begin{equation*}
\abs{\partial_\x \z _\pm (x,y,\x,t)} \geq c_0 (\abs x + \abs y + \abs t)
\end{equation*}
for some $c_0 > 0$. Indeed we have:
\begin{equation*}
\begin {aligned}
\abs{\partial_\x \z (x,y,\x,t)}
 = \abs {\partial_x \vf_\pm(x,\x) - \partial_\x \vf_\pm(y,\x) -2t\x}
 \geq \abs{x-(y+2t\x)} - c R^{-\rho}
\end {aligned}
\end{equation*}
But $(y+2t \x,\x) \in \G_\pm(R_4,d_4,\pm \s_4)$ so if $\abs{x} \geq \g \abs{y + 2t\x}$:
\begin{equation*}
\abs{x-(y+2t\x)} \geq (1 - \g\inv ) \abs x  \geq  \frac {1 - \g \inv}2 (\abs x + \abs {y+2t\x}) \geq c_0 (\abs x + \abs y + \abs t)
\end{equation*}
and if $\abs x \leq \abs{y + 2t\x}$:
\begin{equation*}
\begin{aligned}
\abs{x-(y+2t\x)}
& \geq \innp{x-(y+2t\x)}{{\mp \hat \x}} = \frac {\pm 1} {\abs \x} (\innp {y+2t\x} \x -\innp x \x ) \\
& \geq (\s_5 \abs {y+2t\x} - \s_{4} \abs x ) \geq (\s_5-\s_4) \abs {y+2t\x} \geq c_0 (\abs x + \abs{y+2t\x})\\
& \geq c_0(\abs x + \abs y + \abs t)
\end{aligned}
\end{equation*}
Then we can do partial integrations with the operator $L = \frac {\partial_\x \z . \partial_\x}{\abs {\partial_\x \z}^2}$, each iteration giving a new power of $h$ and $t^{-\rho }$.
\end{proof}

\begin{corollary}
For all $\d \in \R$, there is $\n \in \N$ such that for all $l\in\R$ and $\pm t \geq 0$ we have:
\begin{equation} \label{eqdect3}
\nr{ \pppg x ^l \Op(\o) I_h(b_\pm(t,h),\vf) U_0^h(t) I_h(e_\pm,\vf)^* \pppg x ^{-1- \n -l}} \leq c h^{N+1}  \pppg t ^{-\d}
\end{equation}
\end{corollary}

\begin{proof}
The proof is the same as for \eqref{eqdect} but instead of an estimate of $\nr{I_h(b_{\n,\pm}^j,\vf)}$ we need an estimate of $\nr{\Op(\o)I_h(b_{\n,\pm}^j,\vf)}$. According to lemma 4.4 in \cite{wang88} if we take $R$ large enough, then the supports of $\o(x,\partial_x \vf(x,\x))$ and $b_{\n,\pm}^j$ are disjoint, so this norm is only the norm of the rest given in proposition A.3 of \cite{wang88}. This rest is of order $O(h^{N+1})$ and the time dependance is given as for $\nr{I_h(b_{\n,\pm}^j,\vf)}$ by a finite number of derivatives of $b_{\n,\pm}^j$ so we conclude the same way.
\end{proof}

Now we can prove the main theorem of this section:

\begin{proof}  [Proof of theorem \ref{estimrobert}]

Let $\n \in\N$ given by proposition \ref{propdect} for $\d = 2$. We prove the ``+'' case, and we omit the + subscript for $\vf$, $b$, $p$ and $r$. Let $t\geq 0$. According to \eqref{diff-int} and proposition \ref{propap}, we have:
\begin{equation*}
U_h(t) I_h(b(0,h),\vf) = I_h(b(t,h),\vf) U_0^h(t) -  \int_0^t  U_h(t-s) I_h(p(s,h),\vf) U_0^h(s)\, ds
\end{equation*}
and then, by proposition \ref{propbe}:
\begin{equation*}
\begin{aligned}
U_h(t) \Op(\o_+)
& = h^{N+1} U_h(t) \Op(r(h)) +  I_h(b(t,h),\vf) U_0^h(t) I_h(e(h),\vf)^* \\
& \quad - \int_0^t  U_h(t-s) I_h(p(s,h),\vf) U_0^h(s)I_h(e(h),\vf)^*\, ds\\
\end{aligned}
\end{equation*}
For $\a > \frac 12$ and $\Im z > 0$, using $(\hh-z) \inv  =\frac ih \int_0^\infty e^{\frac {it}h z} U_h(t) \,  dt$ (see theorem 1.10 in \cite{engel}) gives:
\begin{eqnarray*}
\lefteqn{\pppg x ^{-\a} \Op(\o) (\hh-z)\inv \Op(\o_+)\pppg x ^{-\n}}  \\
\nonumber && = h^{N+1} \pppg x ^{-\a} \Op(\o) (\hh-z)\inv \Op(r(h))  \pppg x^{-\n} \\
\nonumber && \quad  + \frac ih \pppg x ^{-\a}\int_{t=0}^\infty e^{\frac{it}h z} \Op(\o) I_h(b(t,h),\vf) U_0^h(t)I_h(e(h),\vf)^*\pppg x ^{-\n}\, dt \\
\nonumber && \quad -  \pppg x ^{-\a}\Op(\o) \int_{s=0}^\infty e^{\frac{is}h z} (\hh-z)\inv  I_h(p(s,h),\vf) U_0^h(s)I_h(e(h),\vf)^*\pppg x ^{-\n}\, ds
\end{eqnarray*}

According to the uniform estimate for the resolvent (see \cite{royer}) the first term is $O(h^N)$. We use \eqref{eqdect3} and \eqref{eqdect2} for the second and third terms, which, after taking the limit $z \to E_h$ if $E_h \in\R$, proves \eqref{estimrob2}. 
\end{proof}

\begin{remark*}
To prove \eqref{estimrob4} we apply the same argument with:
\begin{equation*}
(\hh^* - z)\inv = -\frac ih \int_{-\infty} ^0 e^{-\frac {it} h (\hh^*-z)}\, dt
\end{equation*}
\end{remark*}

\bibliographystyle{alpha}
\bibliography{bibliothese}

\noindent \hspace{2cm} {\sc Julien Royer}

\noindent \hspace{2cm} {\sc Laboratoire de mathématiques Jean Leray}

\noindent \hspace{2cm} {\sc 2, rue de la Houssinière - BP 92208}

\noindent \hspace{2cm} {\sc F-44322 Nantes Cédex 3}

\noindent \hspace{2cm} {\sc France}

\noindent \hspace{2cm} \url{julien.royer@univ-nantes.fr}

\end{document}